\documentclass[manuscript, nonacm]{acmart}

\setcopyright{acmcopyright}
\copyrightyear{}
\acmYear{}

\acmBooktitle{}
\acmPrice{}
\acmISBN{}

\usepackage{graphicx} 
\graphicspath{ {images/} }

\usepackage[scr=rsfs]{mathalpha}

\usepackage{subcaption}
\usepackage{booktabs}
\usepackage[ruled,vlined]{algorithm2e}
\usepackage{hyperref}

\usepackage{amsfonts}
\usepackage{amsmath}
\usepackage{amsthm}

\newtheorem{definition}{Definition}[section]
\newtheorem{theorem}{Theorem}[section]
\newtheorem{lemma}{Lemma}[section]

\newtheorem{observation}{Observation}[section]

\newenvironment{proof-sketch}{\noindent{\bf The proof sketch.}\hspace*{1em}}{\qed\bigskip}

\usepackage{thm-restate}

\newcommand{\sS}{\mathcal{S}} 
\newcommand{\sR}{\mathcal{R}} 
\newcommand{\sP}{\mathcal{P}} 
\newcommand{\rR}{\mathscr{R}} 
\newcommand{\Q}{\mathcal{Q}}
\newcommand{\fR}{\mathfrak{R}} 

\newcommand{\bR}{\mathbb{R}}

\newcommand{\tinyo}{\scriptscriptstyle o}
\newcommand{\domega}{\overset{\tinyo}{\Omega}}
\newcommand{\dtheta}{\overset{\tinyo}{\Theta}}
\newcommand{\dO}{\overset{\tinyo}{O}}

\newcommand{\etal}{\textit{et al.}}

\usepackage{array}
\usepackage{setspace}
\usepackage{hhline}

\newcommand{\ignore}[1]{}

\makeatletter
\DeclareFontFamily{U}{tipa}{}
\DeclareFontShape{U}{tipa}{m}{n}{<->tipa10}{}
\newcommand{\arc@char}{{\usefont{U}{tipa}{m}{n}\symbol{62}}}%

\newcommand{\arc}[1]{\mathpalette\arc@arc{#1}}

\newcommand{\arc@arc}[2]{%
  \sbox0{$\m@th#1#2$}%
  \vbox{
    \hbox{\resizebox{\wd0}{\height}{\arc@char}}
    \nointerlineskip
    \box0
  }%
}
\makeatother

\begin{document}

\title{Lower Bounds for Semialgebraic Range Searching and Stabbing Problems}
\titlenote{2021 SoCG Best Paper. This version contains several improved results comparing to the version in SoCG conference proceedings. Manuscript submitted to JACM.}
\author{Peyman Afshani}
\email{peyman@cs.au.dk}
\author{Pingan Cheng}
\email{pingancheng@cs.au.dk}
\affiliation{
  \institution{Aarhus University}
  \city{Aarhus}
  \country{Denmark}
}
\thanks{Supported by DFF (Det Frie Forskningsr\aa d) of Danish Council for Independent Research under grant ID DFF$-$7014$-$00404.}


\acmDOI{}

\begin{abstract}
In the semialgebraic range searching problem, we are given a set of $n$ points in $\bR^d$ 
and we want to preprocess the points such that for any query range belonging to 
a family of constant complexity semialgebraic sets (Tarski cells), 
all the points intersecting the range can be reported or counted efficiently. 
When the ranges are composed of simplices, then the problem is well-understood:
it can be solved using $S(n)$ space and with $Q(n)$ query time with $S(n)Q(n)^d = \tilde{O}(n^d)$
where the $\tilde{O}(\cdot)$ notation hides polylogarithmic factors and this trade-off is tight
(up to $n^{o(1)}$ factors). 
Consequently, there exists ``low space'' structures that use $O(n)$ space with
$O(n^{1-1/d})$ query time and ``fast query'' structures that use $O(n^d)$ space with
$O(\log^{d+1} n)$ query time. 
However, for the general semialgebraic ranges, 
only ``low space'' solutions are known, but the best solutions match the same trade-off curve
as the simplex queries, with $O(n)$ space and $\tilde{O}(n^{1-1/d})$ query time. 
It has been conjectured that the same could be done for 
the ``fast query'' case but this open problem has stayed unresolved. 

Here, we disprove this conjecture.
We give the first nontrivial lower bounds for semilagebraic range searching
and other related problems.
More precisely, we show that any data structure for reporting the points between two
concentric circles, a problem that we call 2D annulus reporting problem,
with $Q(n)$ query time must use $S(n)=\domega(n^3/Q(n)^5)$ space 
where the $\domega(\cdot)$ notation hides $n^{o(1)}$ factors, meaning,
for $Q(n)=\log^{O(1)}n$, $\domega(n^3)$ space must be used.
In addition, we study the problem of reporting the subset of input points between two polynomials of the form $Y=\sum_{i=0}^\Delta a_i X^i$
where values $a_0, \cdots, a_\Delta$ are given at the query time, a problem that we call polynomial slab reporting.
For this, we show a space lower bound of 
$\domega(n^{\Delta+1}/Q(n)^{(\Delta+3)\Delta/2})$, which shows for  $Q(n)=\log^{O(1)}n$,
we must use $\domega(n^{\Delta+1})$ space. 
We also consider the dual problems of semialgebraic range searching, semialgebraic stabbing problems,
and present lower bounds for them.
In particular, we show that in linear space, any data structure that solves 2D annulus stabbing problems must use $\Omega(n^{2/3})$ query time.
Note that this almost matches the upper bound obtained by lifting 2D annuli to 3D. 
Like semialgebraic range searching, we also present lower bounds for general semialgebraic slab stabbing problems.
Again, our lower bounds are almost tight for linear size data structures in this case.
\end{abstract}


\maketitle

\section{Introduction}
We address one of the biggest open problems of the recent years
in the range searching area.
Our main results are lower bounds in the pointer machine model of computation
that essentially show that the so-called ``fast query'' version of the semialgebraic
range reporting problem is ``impervious'' to the algebraic techniques.
Our main result reveals that to obtain polylogarithmic query time, the data structure
requires $\domega(n^{\Delta+1})$ space\footnote{
$\domega(\cdot)$, $\dO(\cdot)$, $\dtheta(\cdot)$ notations hide $n^{o(1)}$ factors
and $\tilde{\Omega}(\cdot)$, $\tilde{O}(\cdot)$, $\tilde{\Theta}(\cdot)$ notations hide $\log^{O(1)}n$ factors.}, 
where the constant depends on $\Delta$, $n$ is the input size, and $\Delta+1$ is the 
number of parameters of each ``polynomial inequality'' (these will be defined 
more clearly later). 
Thus, we refute a relatively popular recent conjecture that data structures with 
$\dO(n^d)$ space and polylogarithmic query time could exist, where $d$ 
is the dimension of the input points.
Surprisingly, the proofs behind these lower bounds are simple,
and these lower bounds could have been discovered years ago
as the tools we use already existed decades ago.

Range searching is a broad area of research in which 
we are given a set $P$ of $n$ points in $\bR^d$ and the goal is 
to preprocess $P$ such that given a query range $\rR$, we can count
or report the subset of $P$ that lies in $\rR$.
Often $\rR$ is restricted to a fixed family of ranges, e.g.,
in simplex range counting problem, $\rR$ is a simplex in $\bR^d$ and the goal 
is to report $|P \cap \rR|$, or in halfspace range reporting problem, $\rR$ is a halfspace and the goal
is to report $P \cap \rR$. 
Range searching problems have been studied extensively and they have  numerous variants. 
For an overview of this topic, we refer the readers to 
an excellent survey by Agarwal~\cite{toth2017handbook}.

Another highly related problem which can be viewed as the ``dual''  of this problem is range stabbing:
we are given a set $R$ of ranges as input and 
the goal is to preprocess $R$ such that given a query point $p$,
we can count or report the ranges of $R$ containing $p$ efficiently.
Here, we focus on the reporting version of range stabbing problems.

\subsection{Range Searching: A Very Brief Survey}
\subsubsection{Simplex Range Searching}
Simplices is one of the most fundamental family of queries. 
In fact, if the query is decomposable (such as range counting or range reporting queries), 
then simplices can be used as  ``building blocks'' to answer more complicated queries:
for a query $\rR$ which  is a polyhedral region of $O(1)$ complexity, 
we can decompose it into $O(1)$ disjoint simplices (with a constant that depends on $d$)
and thus answering $\rR$ can be reduced to answering $O(1)$ simplicial queries. 

Simplicial queries were hotly investigated in 1980s and this led to development of two important
tools in computational geometry: cuttings and partition theorem and 
both of them have found applications in areas not related to range searching.

\paragraph{Cuttings and Fast Data Structures}
``Fast query'' data structures can answer
simplex range counting or reporting queries in polylogarithmic
query time but by using $O(n^d)$ space and they can be built using cuttings. 
In a nut-shell, given a set $H$ of $n$ hyperplanes in $\bR^d$, a $\frac 1r$-cutting, is a decomposition 
of $\bR^d$ into $O(r^d)$ simplices such that each simplex is intersected by $O(n/r)$ hyperplanes of $H$.
These were developed by some  of the pioneers in the range searching area, such as 
Clarkson~\cite{ClaDCG87},  Haussler and Welzl~\cite{hw87}, Chazelle and Friedman~\cite{Chazelle.Friedman},
Matou\v sek~\cite{Matousek91cuttings}, finally
culminating in a result of Chazelle~\cite{Chazelle.cutting} who optimized various aspects of cuttings.
Using cuttings, 
one can answer simplex range counting, or 
reporting queries with $O(n^d)$ space and $O((\log n)^{d+1} + k)$ query time (where $k$ is
the output size)~\cite{matouvsek1993range}.
The query time can be lowered to $O(\log n)$ by increasing the space slightly to $O(n^{d+\varepsilon})$ for
any constant $\varepsilon >0$~\cite{chazelle1989quasi}.
An interested reader can refer to a survey on cuttings by Chazelle~\cite{Chazelle.book}.

\paragraph{The Partition Theorem and Space-efficient Data Structures}
At the opposite end of the spectrum,
simplex range counting or reporting queries can be answered using linear space but with higher
query time of $O(n^{1-1/d})$, using partition trees and the related techniques. 
This branch of techniques has a very interesting history. 
In 1982, Willard~\cite{Willardpartition} cleverly used ham sandwich theorem
to obtain
a linear-sized data structure with query time of $O(n^{\gamma})$ for some constant $\gamma<1$ for simplicial queries in 2D. 
After a number of attempts that either improved the exponent or generalized the technique
to higher dimensions, Welzl~\cite{Welzlpartree82} in 1982 provided the first optimal 
exponent for the partition trees, then Chazelle \etal~\cite{chazelle1989quasi} provided
the first near-linear size data structure with query time of roughly
$O(n^{1-1/d})$. 
Finally, a data structure with $O(n)$ space and $O(n^{1-1/d})$ query time was given
by Matou\v sek~\cite{matouvsek1993range}. This was also simplified recently by
Chan~\cite{chan2012optimal}.

\paragraph{Space/Query Time Trade-off}
It is possible to combine fast query data structures and linear-sized data structures
to solve simplex queries with $S(n)$ space and $Q(n)$ query time such that
$S(n)Q(n)^d = \tilde{O}(n^d)$.
This trade-off between space and query time is optimal, at least in the pointer machine
model and in the semigroup model~\cite{afshani2012improved,chazelle1996simplex,chazelle1989lower}.

\paragraph{Multi-level Structures, Stabbing and Other Related Queries}
By using multi-level data structures, one can solve more complicated problems where
both the input and the query shapes can be simplicial objects of constant complexity.
The best multi-level data structures use one extra $\log n$ factor in space and query time per 
level~\cite{chan2012optimal} and there exist lower bounds that show 
space/query time trade-off should blow up by at least
$\log n$ factor per level~\cite{AD.frechet17}.
This means that problems such as simplex stabbing (where the input is a set of simplices
and we want to output the simplices containing a given query point) or simplex-simplex containment 
problem (where the input is a set of simplices, and we want to output simplices fully
contained in a query simplex) all have the same trade-off curve of $S(n)Q(n)^d = \tilde{O}(n^d)$
between space $S(n)$ and query time $Q(n)$.

Thus, one can see that the simplex range searching as well as its generalization 
to problems where both the input and the query ranges are  ``flat'' objects is very well understood. 
However, there are many natural query ranges that cannot be represented using simplices, e.g.,
when query ranges are spheres in $\bR^d$.
This takes us to semialgebraic range searching. 

\subsubsection{Semialgebraic Range Searching}
A semialgebraic set is defined as a subset of $\bR^d$ that can
be described as the union or intersection of $O(1)$ ranges, where each
range is defined by $d$-variate polynomial inequality of degree at most $\Delta$, defined
by at most $B$ values given at the query time; we call $B$ the \textit{parametric dimension}.
For instance, with $B=3$, $\Delta=2$, and given three values $a,b$ and $c$ at the query time, 
a circular query can be represented as 
$\left\{ (X,Y) \in \bR^2| (X-a)^2 + (Y-b)^2 \le c^2  \right\}$.
In semialgebraic range searching, the queries are semialgebraic sets.

Before the recent ``polynomial method revolution'', the tools available to deal with 
semialgebraic range searching were limited, at least compared to the simplex queries. 
One way to deal with semialgebraic range searching is through linearization~\cite{YaoYaolinearization}.
This idea maps the input points to $\bR^L$, for some potentially large parameter $L$, such that
each polynomial inequality can be represented as a halfspace.
Consequently, semialgebraic range searching can be solved with the space/query time trade off
of $S(n)Q(n)^L = \tilde{O}(n^L)$.
The exponent of $Q(n)$ in the trade-off can be improved (increased) a bit by exploiting that in $\bR^L$, the
input set actually lies in a $d$-dimensional surface~\cite{agarwal1994range}.
It is also possible to build ``fast query'' data structures but using
$O(n^{B+\varepsilon})$ space by a recent result of Agarwal~\etal~\cite{agarwal2019an}.

In 2009, Zeev Dvir~\cite{Dvirkakeya} proved the discrete Kakeya problem with a very elegant and simple
proof, using a polynomial method. Within a few years, this led to revolution in discrete and
computational geometry, one that was ushered in by Katz and Guth's almost tight bound on 
Erd\H os distinct distances problem~\cite{guth2015erdHos}. For a while, the polynomial method
did not have much algorithmic consequences but this changed with the work of 
Agarwal, Matou\v sek, and Sharir~\cite{agarwal2013range} where they showed that
at least as long as linear-space data structures are considered, semialgebraic range queries 
can essentially be solved within the same time as simplex queries (ignoring some lower order terms).
Later developments (and simplifications) of their approach by Matou\v sek and 
Pat\' akov\' a~\cite{MatousekZuzana} lead to the current best results:
a data structure with linear size and with query time of $\tilde{O}(n^{1-1/d})$.

\paragraph{Fast Queries for Semialgebraic Range Searching: an Open Problem}
Nonetheless, despite the breakthrough results brought on by the algebraic techniques, the fast
query case still remained unsolved, even in the plane: e.g., the best known data structures for 
answering circular queries with polylogarithmic query time still use $\tilde{O}(n^3)$ space,
by using linearization to $\bR^3$.
The fast query case of semialgebraic range searching has been explicitly
mentioned as a major open problem in multiple recent
publications\footnote{
To quote Agarwal \etal~\cite{agarwal2013range},``[a] very interesting and challenging problem is, in our opinion, the fast-query
case of range searching with constant-complexity semialgebraic sets, where the goal
is to answer a query in $O(\log n)$ time using roughly $O(n^d)$ space.''
The same conjecture is repeated in a different survey~\cite{Agarwal2017} and it is also 
  emphasized that the question is even open for disks in the plane, ``... whether a
  disk range-counting query in $\bR^2$ be answered in $O(\log n)$ time using
$O( n^2)$ space?''.}.
In light of the breakthrough result of Agarwal \etal~\cite{agarwal2013range}, it is quite reasonable
to conjecture that semialgebraic range searching should have
the same trade-off curve of $S(n)Q(n)^d = \tilde{O}(n^d)$.

Nonetheless, the algebraic techniques have failed to make sufficient advances to settle this open problem.
As mentioned before, the best known ``fast query'' result by Agarwal~\etal~\cite{agarwal2019an} uses $O(n^{B+\epsilon})$ space.
In general, $B$ can be much larger than $d$ and thus it leaves a big gap between
current best upper bound and the conjectured one.
Given that it took a revolution caused by the polynomial method to advance our knowledge
of the ``low space'' case of semialgebraic range searching, 
it is not too outrageous to imagine that perhaps equally revolutionary techniques are needed
to settle the ``fast query'' case of semialgebraic range searching. 
 
 \subsubsection{Semialgebraic Range Stabbing}
Another important problem is semialgebraic stabbing, 
where the input is a set of $n$ semialgebraic sets, i.e., ``ranges'', and queries are points.
The goal is to output the input ranges that contain a query point. 
Here, ``fast query'' data structures are possible,
for example by observing that an arrangement of $n$ disks in the plane has $O(n^2)$ complexity,
counting or reporting the disks stabbed by a query point can be done with $O(n^2)$ space and $O(\log n)$ query time
using slab and persistent.
However, this simple idea does not work for general semialgebraic sets in higher dimensions.
The main bottleneck before the recent ``polynomial method revolution''
is that for dimensions higher than $4$ no technique was able to decompose 
the arrangement formed by a collection of algebraic surfaces 
into an arrangement with the property that each cell is a constant-complexity semialgebraic set
and the complexity of the decomposed arrangement is close to that of the original arrangement.
This was also one of the roadblocks for the ``fast query'' version of semialgebraic range searching.
Recently, Guth~\cite{guth2015gen} showed
the existence of polynomials whose connected components
give such decompositions for arbitrary dimensions.
An efficient algorithm for computing such partitioning polynomials was given very recently by Agarwal \etal~\cite{agarwal2019an},
and as one of the results, they showed it is possible to build ``fast query''
data structures for semialgebraic stabbing problems using $O(n^{d+\varepsilon})$ space
for $n$ semialgebraic sets in $\bR^d$.
However, comparing to the reporting problems, in this stabbing scenario,
it seems difficult to make advancements in the 
``low space'' side of things; 
e.g., for the planar disk stabbing problem, 
the only known data structure with $O(n)$ space is one that uses linearization to 3D that
results in $\tilde{O}(n^{2/3})$ query time.

\subsection{Our Results}
Our main results are lower bounds in the pointer machine model of computation for
four central problems defined below.
In the \textit{2D polynomial slab reporting} problem,
given a set $\sP$ of $n$ points in $\bR^2$,
the task is to preprocess $\sP$
such that given a query 2D polynomial slab $\rR$,
the points contained in the polynomial slab, i.e., $\rR \cap \sP$, can be reported efficiently.
Informally, a 2D polynomial slab is the set of points $(x,y)$ such that $P(x) \le y \le P(x)+w$,
for some univariate polynomial $P(x)$ of degree $\Delta$ and value $w$ given at the query time. 
In the \textit{2D polynomial slab stabbing} problem, the input is a set of $n$ polynomial slabs
and the query is a point $q$ and the goal is to report all the slabs that contain $q$.
Similarly, in the \textit{2D annulus reporting} problem, the input is a set $P$ of $n$ points in $\bR^2$
and the query is an ``annulus'', the region between two concentric circles. 
Finally, in \textit{2D annulus stabbing} problem, the input is a set of $n$ annuli, the query is a
point $q$ and the goal is to report all the annuli that contain $q$. 

For polynomial slab queries, 
we show that if a data structure answers queries in $Q(n) + O(k)$ time, where $k$ is the output size,
using $S(n)$ space, then 
$S(n)=\domega(n^{\Delta+1}/Q(n)^{(\Delta+3)\Delta/2})$;
the hidden constants depend on $\Delta$.
So  for ``fast queries'', i.e., $Q(n) = \tilde{O}(1)$, 
$\domega(n^{\Delta+1})$ space must be used.
This is
\textit{almost tight} as the exponent matches the upper bounds obtained by linearization
as well as the recent upper bound of Agarwal~\etal~\cite{agarwal2019an}!
Also, we prove that
any structure that answers polynomial slab stabbing queries in $Q(n) + O(k)$ time 
must use $\Omega(n^{1+2/(\Delta+1)}/Q(n)^{2/\Delta})$ space.
In the ``low space'' setting, when $S(n) = O(n)$, this gives
$Q(n) = \Omega(n^{1-1/(\Delta+1)})$.
This is once again \textit{almost tight}, as it matches the upper bounds
obtained by linearization for when $S(n) = \tilde{O}(n)$.

For the annulus reporting problem, we get the same
$S(n)=\domega(n^{3}/Q(n)^{5})$ lower bound as polynomial slab reporting when $\Delta=2$.
For the annulus stabbing problem, we show $S(n)=\Omega(n^{3/2}/Q(n)^{3/4})$, e.g.,
in ``low space'' setting when $S(n) = O(n)$, we must have
$Q(n) = \Omega(n^{2/3})$; compare this with simplex stabbing queries can be solved with $O(n)$ space and
$\tilde{O}(\sqrt{n})$ query time. 
As before, this is almost tight, as it matches the upper bounds obtained by linearization to 3D for when $S(n)=\tilde{O}(n)$.

Somewhat disappointedly, no revolutionary new technique is required to obtain these results.
We use novel ideas in the construction of ``hard input instances'' but otherwise
we use the two widely used pointer machine lower bound frameworks
by Chazelle~\cite{chazelle1990lower},  Chazelle and Rosenberg~\cite{chazelle1996simplex},
and Afshani~\cite{afshani2012improved}.
Our results are summarized in Table~\ref{tab:results}.

{\small
\begin{table}[H]
\setlength\extrarowheight{5pt}
\begin{minipage}{\textwidth}
\centering
\caption{Our Results, $*$ indicates this paper. 
In the table, $\domega(\cdot)$ and $\dO(\cdot)$ notations hide $n^{o(1)}$ factors,
and $\tilde{O}(\cdot)$ notation hides $\log^{O(1)}n$ factors.}
\label{tab:results}
\begin{tabular}{ | m{4.4cm} | m{4cm}| m{4.3cm} |} 
\hline
\bf{Problem}& \bf{Lower Bound} & \bf{Upper Bound} \\
\hline
2D Polynomial Slab Reporting  & $S(n)=\domega\left(\frac{n^{\Delta+1}}{Q(n)^{(\Delta+3)\Delta/2}}\right)^*$ & 
$S(n)=\tilde{O}\left(\frac{n^{\Delta+1}}{Q(n)^{2\Delta}}\right)$~\cite{agarwal1994range, agarwal2013range, matouvsek1993range, agarwal2019an} \\
\bf{When} $\boldsymbol{Q(n)=\dO(1)}$ & $\boldsymbol{S(n)=\domega\left(n^{\Delta+1}\right)}^*$ & 
$\boldsymbol{S(n)=\dO\left(n^{\Delta+1}\right)}$~\cite{agarwal1994range, agarwal2013range, matouvsek1993range, agarwal2019an}  \\
\hline
2D Annulus Reporting  & $S(n)=\domega\left(\frac{n^3}{Q(n)^{5}}\right)^*$ & $S(n)=\tilde{O}\left(\frac{n^3}{Q(n)^4}\right)$
~\cite{agarwal1994range, agarwal2013range, matouvsek1993range, agarwal2019an}\\
\bf{When} $\boldsymbol{Q(n)=\dO(1)}$ & $\boldsymbol{S(n)=\domega\left(n^3\right)}^*$ & 
$\boldsymbol{S(n)=\dO\left(n^3\right)}$~\cite{agarwal1994range, agarwal2013range, matouvsek1993range, agarwal2019an}  \\
\hhline{===}
2D Polynomial Slab Stabbing  & 
$S(n)=\Omega\left(\frac{n^{1 + 2/(\Delta+1) }}{Q(n)^{ 2/\Delta }}\right)^*$ &
$S(n)=\tilde{O}\left(\frac{n^2}{Q(n)^{(\Delta+1)/\Delta}}\right)$~\cite{agarwal2019an}\\
\bf{When} $\boldsymbol{S(n)=\dO(n)}$ & $\boldsymbol{Q(n)=\domega\left(n^{1-1/(\Delta+1)}\right)}^*$ & $\boldsymbol{S(n)=\dO\left(n^{1-1/(\Delta+1)}\right)}$~\cite{agarwal2019an}\\
\hline
2D Annulus Stabbing & $S(n)=\Omega\left(\frac{n^{3/2}}{Q(n)^{3/4}}\right)^*$ & 
$S(n)=\tilde{O}\left(\frac{n^2}{Q(n)^{3/2}}\right)$~\cite{agarwal2019an}\\
\bf{When} $\boldsymbol{S(n)=\dO(n)}$ & $\boldsymbol{Q(n)=\domega\left(n^{2/3}\right)}^*$ & $\boldsymbol{Q(n)=\dO\left(n^{2/3}\right)}$~\cite{agarwal2019an}\\
\hline
\end{tabular}
\end{minipage}
\end{table}
}

\section{Preliminaries}

We first review the related geometric reporting data structure lower bound frameworks.
The model of computation we consider is (an augmented version of) the pointer machine model.

In this model,
the data structure is a directed graph $M$.
Let $\sS$ be the set of input elements.
Each cell of $M$ stores an element of $\sS$
and two pointers to other cells.
Assume a query $q$ requires a subset 
$\sS_q\subset \sS$ to be output.
For the query, we only charge for the pointer navigations.
Let $M_q$ be the smallest connected subgraph, s.t.,
every element of $\sS_q$ is stored in at least one element of $M_q$.
Clearly, $|M|$ is a lower bound for space and $|M_q|$ is a lower bound for query time. 
Note that this grants the algorithm  unlimited computational power as well as
full information about the structure of $M$.

In this model, 
there are two main lower bound frameworks,
one for range reporting~\cite{chazelle1990lower,chazelle1996simplex}, 
and the other for its dual, range stabbing~\cite{afshani2012improved}.
We describe them in detail here.

\subsection{A Lower Bound Framework for Range Reporting Problems}

The following result by Chazelle~\cite{chazelle1990lower} and later Chazelle and Rosenberg~\cite{chazelle1996simplex} 
provides a general lower bound framework for range reporting problems.
In the problem, we are given a set $\sS$ of $n$ points in $\bR^d$
and the queries are from a set $\sR$ of ranges.
The task is to build a data structure such that given any query range $\rR \in \sR$,
we can report the points intersecting the range, i.e., $\rR \cap \sS$, efficiently.

\begin{theorem}[Chazelle~\cite{chazelle1990lower} and Chazelle and Rosenberg~\cite{chazelle1996simplex}]
\label{thm:chazelle-framework}
Suppose there is a data structure for range reporting problems that uses at most $S(n)$ space
and can answer any query in $Q(n)+O(k)$ time where $n$ is the input size and $k$ is the output size.
Assume we can show that there exists an input set $\sS$ of $n$ points satisfying the following:
There exist $m$ subsets $q_1, q_2,\cdots, q_m\subset \sS$,
where $q_i, i=1,\cdots,m$, is the output of some query
and they satisfy the following two conditions:
(i) for all $i=1,\cdots,m$, $|q_i|\ge Q(n)$;
and (ii) the size of the intersection of every $\alpha$ distinct subsets $q_{i_1},q_{i_2},\cdots,q_{i_\alpha}$
is bounded by some value $c\ge 2$, i.e., $|q_{i_1}\cap q_{i_2}\cap \cdots \cap q_{i_\alpha}|\le c$.
Then $S(n)=\Omega(\frac{\sum_{i=1}^m|q_i|}{\alpha2^{O(c)}})=\Omega(\frac{mQ(n)}{\alpha2^{O(c)}})$.
\end{theorem}

To use this framework, we need to exploit the property of the considered problem 
and come up with a construction that satisfies the two conditions above.
Often, the construction is randomized and thus one challenge is to satisfy condition (ii) in the worst-case.
This can be done by showing that the probability that (ii) is violated is very small and then using
a union bound to prove that with positive probability the construction satisfies (ii) in the worst-case. 

\subsection{A Lower Bound Framework for Range Stabbing Problems}
Range stabbing problems can be viewed as the dual of range reporting problems.
In this problem, we are given a set $\sR$ of $n$ ranges, 
and the queries are from a set $\Q$ of $n$ points. 
The task is to build a data structure such that given any query point $q\in\Q$,
we can report the ranges ``stabbed'' by this query point,
i.e., $\{\rR \in \sR: \rR \cap q \neq \emptyset\}$, efficiently.
A recent framework by Afshani~\cite{afshani2012improved} provides a simple
way to get the lower bound of such problems.

\begin{theorem}[Afshani~\cite{afshani2012improved}]
\label{thm:afshani-framework}
Suppose there is a data structure for range stabbing problems that uses at most $S(n)$ space
and can answer any query in $Q(n)+O(k)$ time where $n$ is the input size and $k$ is the output size.
Assume  we can show that there exists an input set $R \subset \sR$ of $n$ ranges that satisfy the following:
(i) every query point of the unit  square $U$ is contained in at least $t\ge Q(n)$ ranges;
and (ii) the area of the intersection of every $\alpha<t$ ranges is at most $v$.
Then $S(n)=\Omega(\frac{t}{v2^{O(\alpha)}})=\Omega(\frac{Q(n)}{v2^{O(\alpha)}})$.
\end{theorem}

This is very similar to framework of Theorem~\ref{thm:chazelle-framework} but often it requires no
derandomization.
\ignore{
This lower bound framework has been used to prove lower bounds for
slab enclosure, halfspace range reporting, simplex range reporting~\cite{afshani2012improved},
orthogonal range reporting~\cite{afshani2010orthogonal}, 
rectangle stabbing~\cite{afshani2012higher,afshani20202d} and other related problems.
}

\subsection{Derandomization Lemmas}
As mentioned before, one challenge of using Chazelle's framework
is to show the existence of a hard instance in the worst case
while the construction is randomized.
To do that, we present two derandomization lemmas.

\begin{lemma}
\label{lem:derand-int}
Let $\sP$ be a set of $n$ points chosen uniformly at random in a square $S$ of side length $n$ in $\bR^2$.
Let $\sR$ be a set of ranges in $S$ such that 
(i) the intersection area of any $t\ge 2$ ranges $\rR_1, \rR_2, \cdots \rR_t \in \sR$ 
is bounded by $O\left(n/2^{\sqrt{\log n}}\right)$;
(ii) the total number of intersections is bounded by $O\left(n^{2k}\right)$ for $k\ge 1$.
Then with probability $>\frac{1}{2}$, for all distinct ranges $\rR_1, \rR_2, \cdots, \rR_t \in \sR$,
$|\rR_1\cap \rR_2 \cap \cdots \rR_t \cap \sP| < 3k\sqrt{\log n}$.
\end{lemma}

\begin{proof}
Consider any intersection region $\rho\in S$ of $t$ ranges with area $A$.
Let $X$ be an indicator random variable with
\[
X_i=
\begin{cases}
1, \textrm{the $i$-th point is inside $\rho$},\\
0, \textrm{otherwise.}
\end{cases}
\]
Let $X=\sum_{i=1}^nX_i$.
Clearly, $\mathbb{E}[X] = \frac{A}{n}$.
By Chernoff's bound,
\[
\Pr\left[X\ge(1+\delta)\frac{A}{n}\right]<\left(\frac{e^\delta}{(1+\delta)^{1+\delta}}\right)^{\frac{A}{n}},
\]
for any $\delta>0$.
Let $\tau=(1+\delta)\frac{A}{n}$, then
\[
\Pr[X\ge\tau]<\frac{e^{\delta\frac{A}{n}}}{(1+\delta)^\tau}<\frac{e^\tau}{(1+\delta)^\tau}=\left(\frac{eA}{n\tau}\right)^\tau.
\]

Now we pick $\tau =3k\sqrt{\log n}$, since $A\le cn/2^{\sqrt{\log n}}$ for some constant $c$, we have
\[
\Pr\left[X\ge3k\sqrt{\log n}\right]<\left(\frac{ce}{2^{\sqrt{\log n}}3k\sqrt{\log n}}\right)^{3k\sqrt{\log n}}<\frac{(ce)^{3k\sqrt{\log n}}}{n^{3k}}.
\]
Since the total number of intersections is bounded by $O(n^{2k})$, the number of cells in the arrangement is also bounded
by $O(n^{2k})$ and thus 
by the union bound, for sufficiently large $n$,
with probability $>\frac{1}{2}$,
the number of points in every intersection region is less than $3k\sqrt{\log n}$.
\end{proof}

\begin{lemma}
\label{lem:derand-ring}
Let $\sP$ be a set of $n$ points chosen uniformly at random in a square $S$ of side length $n$ in $\bR^2$.
Let $\sR$ be a set of ranges in $S$ such that 
(i) the intersection area of any range $\rR \in \sR$ and $S$ is at least $cnt$ 
for some constant $c\ge4k$ and a parameter $t\ge \log n$, where $k \ge 2$;
(ii) the total number of ranges is bounded by $O\left(n^{k+1}\right)$.
Then with probability $>\frac{1}{2}$, for every range $\rR \in \sR$,
$|\rR \cap \sP| \ge t$.
\end{lemma}

\begin{proof}
The proof of this lemma is similar to the one for Lemma~\ref{lem:derand-int}.
We pick $n$ points in $S$ uniformly at random.
Let $X_{ij}$ be the indicator random variable with
\[
X_{ij}=
\begin{cases}
1, \textrm{point $i$ is in range $j$},\\
0, \textrm{otherwise.}
\end{cases}
\]
We know that the area of each range is at least $cnt$.
Then the expected number of points in each range is $ct$.
Consider an arbitrary range, let $X_j=\sum_{i=1}^nX_{ij}$, then by Chernoff's bound
\begin{align*}
\Pr&\left[X_j<\left(1-\frac{c-1}{c}\right)ct\right]<e^{-\frac{\left(\frac{c-1}{c}\right)^2ct}{2}}\\
\implies\Pr&[X_j<t]<e^{-\frac{(c-1)^2t}{2c}}<\frac{1}{n^{\frac{(c-1)^2}{2c}}}\le\frac{1}{n^{2k-1+1/(8k)}}.
\end{align*}
The second last inequality follows from $t\ge \log n$ and the last inequality follows from $c\ge4k$.
Since the total number of ranges is bounded by $O(n^{k+1})$,
by a standard union bound argument, the lemma holds.
\end{proof}

\section{2D Polynomial Slab Reporting and Stabbing}
We first consider the case when query ranges are 2D polynomial slabs.
The formal definition of 2D polynomial slabs is as follows.

\begin{definition}
Let $P(x)=\sum_{i=0}^\Delta a_ix^i$, where $a_\Delta \neq 0$, be a degree $\Delta$ univariate polynomial.
A 2D polynomial slab is a pair $(P(x), w)$, where $P(x)$ is called the base polynomial
and $w>0$ the width of the polynomial slab.
The polynomial slab is then defined as  
$\{(x,y)\in\bR^2:P(x)\le y\le P(x)+w\}$.
\end{definition}

\subsection{2D Polynomial Slab Reporting}
We consider the 2D polynomial slab reporting problem in this section, 
where  the input is a set $\sP$ of $n$ points in $\bR^2$, and the query is a polynomial slab. 
This is an instance of semialgebraic range searching where we have
two polynomial inequalities where 
each inequality has degree $\Delta$ and it is defined by $\Delta+1$ parameters given at the query time
(thus, $B=\Delta+1$).
Note that $\Delta+1$ is also the dimension of linearization for this problem,
meaning, the 2D polynomial slab reporting problem can be lifted to the simplex range
reporting problem in $\bR^{\Delta+1}$.
Our main result shows that for fast queries (i.e., when the query time is polylogarithmic),
this is tight, by showing an $\domega(n^{\Delta+1})$ space lower bound, in the
pointer machine model of computation. 

Before we present the lower bound, we first introduce a simple property of polynomials,
which we will use to upper bound the intersection area of polynomials.
Given a univariate polynomial $P(x)$, 
the following simple lemma establishes the relationship between the leading coefficient
and the maximum range within which its value is bounded.

\begin{lemma}
\label{lem:poly-int}
Let $P(x)=\sum_{i=0}^\Delta a_ix^i$ be a degree $\Delta$ univariate polynomial  
where $\Delta>0$ and $\left|a_\Delta\right|\ge d$ for some positive $d$.
Let $w$ be any positive value and $x_l$ be a parameter. 
If $|P(x)|\le w$ for all $x\in[x_l,x_l+t]$, then $t\le (\Delta+1)^3\left(\frac{w}{d}\right)^{\frac{1}{\Delta}}$.
\end{lemma}

\begin{proof}
First note that w.l.o.g., we can assume $x_l=0$, 
because otherwise we can consider a new polynomial $P'(x)=P(x+x_l)$.
Since $P'(x)$ is still a degree $\Delta$ univariate polynomial with $\left|a_{\Delta}\right| \ge d$,
and for all $x\in [0, t]$, $P'(x) = P(x+x_l)$,
to bound $t$, we only need to consider $P'(x)$ on interval $[0, t]$.

Assume for the sake of contradiction that $t>(\Delta+1)^3\left(\frac{w}{d}\right)^{\frac{1}{\Delta}}$.
We show that this will lead to $\left|a_\Delta\right| < d$.

We pick $\Delta + 1$ different points $(x_i,y_i)$, where $x_i \in [0, t]$ and $y_i=P(x_i)$, on the polynomial.
Then $P(x)$ can be expressed as
\[
P(x)=\sum_{i=0}^{\Delta} y_i\prod_{j=0, j \neq i}^{\Delta} \frac{x-x_j}{x_i-x_j}.
\]
The coefficient of the degree $\Delta$ term is therefore
\[
a_{\Delta} = \sum_{i=0}^{\Delta}y_i\prod_{j=0,j \neq i}^{\Delta}\frac{1}{x_i-x_j}.
\]
We pick $x_i=\frac{t}{i+1}$ for $i = 0, 1, \cdots, \Delta$ and we therefore obtain
\[
a_{\Delta}=\sum_{i=0}^{\Delta}y_i\prod_{j=0,j\neq i}^{\Delta}\frac{1}{\frac{t}{i+1}-\frac{t}{j+1}}.
\]
We now upper bound $\left|\frac{1}{\frac{t}{i+1}-\frac{t}{j+1}}\right|$.
We assume $i < j$, the case for $i > j$ is symmetric.
When $i < j$,
\[
\left|\frac{1}{\frac{t}{i+1}-\frac{t}{j+1}}\right| = \frac{1}{\frac{t}{i+1}-\frac{t}{j+1}}=\frac{(i+1)(j+1)}{t(j-i)} < \frac{(\Delta+1)^2}{t},
\]
where the last inequality follows from $i, j = 0, 1, \cdots, \Delta$ and $j-i \ge 1$.
Also by assumption, $|y_i|\le w$, we therefore have
\[
\left|a_{\Delta}\right| < \frac{w(\Delta+1)(\Delta+1)^{2\Delta}}{t^{\Delta}}<d,
\]
where the last inequality follows from $t>(\Delta+1)^3\left(\frac{w}{d}\right)^{\frac{1}{\Delta}}$.
However, in $P(x)$, $\left|a_{\Delta}\right| \ge d$, a contradiction.
Therefore, $t \le (\Delta+1)^3\left(\frac{w}{d}\right)^{\frac{1}{\Delta}}$.
\end{proof}


With Lemma~\ref{lem:poly-int} at hand,
we now show a lower bound for polynomial slab reporting.

\begin{theorem}
\label{thm:general-ring-lb}
Let $\sP$ be a set of $n$ points in $\bR^2$.
Let $\sR$ be the set of all 2D polynomial slabs $\{(P(x),w):\deg(P)=\Delta \ge 2, w>0\}$.
Then any data structure for $\sP$ that solves polynomial slab reporting for queries from $\sR$
with query time $Q(n)+O(k)$, where $k$ is the output size,
uses $S(n)=\domega\left({n^{\Delta+1}}/{Q(n)^{(\Delta+3)\Delta/2}}\right)$ space.
\end{theorem}

\begin{proof}
We use Chazelle's framework to prove this theorem.
To this end, we will need to show the existence of a hard input instance.
We do this as follows.
In a square $S$, we construct a set of special polynomial slabs
with the following properties:
(i) The intersection area of any two slabs is small; 
and (ii) The area of each slab inside $S$ is relatively large.
Intuitively and consequently, if we sample $n$ points uniformly at random in $S$,
in expectation, few points will be in the intersection of two slabs,
and many points will be in each slab.
Intuitively, this satisfies the two conditions of Theorem~\ref{thm:chazelle-framework}.
By picking parameters carefully and a derandomization process, we get our theorem.
Next, we describe the details.

Consider a square $S=[0,n] \times [0,n]$.
Let $d_i$ for $i=1,2,\cdots,\Delta$ and $w$ be some parameters to be  specified later.
We generate a set of $\Theta\left(\frac{n^{\Delta}}{2^{\Delta}\prod_{i=1}^{\Delta}d_i} \cdot \frac{n}{w}\right)$ 
polynomial slabs $(P(x), w)$ with
\[
P(x)=\left(\sum_{i=1}^\Delta\frac{j_id_ix^i}{n^i}\right)+kw
\]
where $j_i = \lfloor\frac{n}{2d_i}\rfloor, \lfloor\frac{n}{2d_i}\rfloor+1, \cdots, \lfloor\frac{n}{d_i}\rfloor$
for $1 \le i \le  \Delta$ and 
$k = \lfloor\frac{n}{4w}\rfloor, \lfloor\frac{n}{4w}\rfloor+1, \cdots, \lfloor\frac{n}{2w}\rfloor$.
Note that we normalize the coefficients such that
for any polynomial slab in range  $x\in[0,n]$,
a quarter of this slab is contained in $S$ if $w < n/6$.
To show this, it is sufficient to show that every polynomial is inside $S$,
for every $x \in [0,n/4]$.
As all the coefficients of the polynomials are positive, it is sufficient to upper bound 
$P(n/4)$, among all the polynomials $P(x)$ that we have generated.
Similarly, this maximum is attained when all the coefficients are set to their maximum value, 
i.e., when $j_i = n/d_i$ and $k=n/(2w)$, resulting in the polynomial
$P_u(x)=\left(\sum_{i=1}^{\Delta}{x^i}/{n^{i-1}}\right)+\frac{n}{2}$.
Now it easily follows that $P_u(n/4) < 5n/6$.
Then, the claim follows from the following simple observation.
\begin{observation}\label{ob:int}
  The area of a polynomial slab $\left\{ P(x),w \right\}$ for when $a \le x \le b$ is $(b-a)w$.
\end{observation}
\begin{proof}
  The claimed area is $(\int_a^b (P(x)+w) dx) - (\int_a^b P(x) dx) = \int_{a}^b w dx = (b-a)w$.
\end{proof}

Next, we bound the area of the intersection of two polynomial slabs. 
Consider two distinct slabs $\rR_p = (P(x), w)$ and $\rR_q = (Q(x),w)$.
Observe that by our construction, if $P(x)$ and $Q(x)$ only differ in their constant terms,
their intersection is empty.
So we only consider the case that there exists some $0<i\le\Delta$,
such that the coefficients for $x^i$ are different in $P(x)$ and $Q(x)$.
As each slab is created using two polynomials of degree $\Delta$, 
$\rR_q \cap \rR_p$ can have at most $O(\Delta)$ connected regions. 
Consider one connected region $\fR$ and let the interval $\eta = [x_1, x_2] \subset [0, n]$,
be the projection of $\fR$ onto the $X$-axis. 
Define the polynomial $R(x) = P(x) - Q(x)$ and observe that we must have
$|R(x)| \le w$ for all $x \in [x_1, x_2]$. 
We now consider the coefficient of the highest degree term of $R(x)$. 
Let $j_i d_i / n^i$ (resp. $j'_i d_i / n^i$) be the coefficient of the degree $i$ term in $P(x)$ (resp. $Q(x)$). 
Clearly, if $j_i = j'_i$, then the coefficient of $x^i$ in $R(x)$ will be zero.
Thus, to find the highest degree term in $R(x)$, we need to consider the largest index $i$
such that $j_i \not = j'_i$; in this case, $R(x)$ will have degree $i$ and coefficient of $x_i$ will
have absolute value $\left|(j_i - j'_i) d_i/n^i\right| \ge d_i/n^i$.
By Lemma~\ref{lem:poly-int},
$x_2 - x_1 \le O(\Delta^3)\left( \frac{w n^i}{d_i}\right)^{1/i}$.
Next, by Observation~\ref{ob:int}, 
the area of the intersection of $\rR_q$ and $\rR_p$ is $O(\Delta^3)nw \left( \frac{w}{d_i} \right)^{1/i}$.

We pick $d_i = c\Delta^{3i} w^{i+1}2^{i\sqrt{\log n}}$ and $w=16\Delta Q(n)$,
for a large enough constant $c$.
Then, the intersection area of any two polynomial slabs is bounded by $n/2^{\sqrt{\log n}}$.
Since in total we have generated $O(n^{\Delta + 1})$ slabs,
the total number of intersections they can form is bounded by $O(n^{2(\Delta+1)})$.
By Lemma~\ref{lem:derand-int}, with probability $>\frac{1}{2}$,
the number of points of $\sP$ in any intersection of two polynomial slabs is at most $3(\Delta+1)\sqrt{\log n}$.
Also, as we have shown that 
the intersection area of every slab with $S$ 
is at least $nw/4=4 \Delta nQ(n)$,
by Lemma~\ref{lem:derand-ring},
with probability more than $\frac{1}{2}$,
each polynomial slab has at least $Q(n)$ points of $\sP$.

It thus follows that with positive probability, both conditions of Theorem~\ref{thm:chazelle-framework} are satisfied, 
and consequently, we obtain the lower bound of 
\[
S(n)=\Omega\left(\frac{Q(n)\cdot\frac{n^\Delta}{2^{\Delta}\prod_{i=1}^{\Delta}d_i}\cdot\frac{n}{w}}{2^{3(\Delta+1)\sqrt{\log n}}}\right)
=\domega\left(\frac{n^{\Delta+1}}{Q(n)^{(\Delta+3)\Delta/2}}\right),
\]
where the last equality follows from 
$d_i = c\Delta^{3i} w^{i+1}2^{i\sqrt{\log n}}$, $w=16\Delta Q(n)$, and
\[
\prod_{i=1}^{\Delta}d_i=c\Delta^{3(1+\Delta)\Delta/2}w^{(2+\Delta+1)\Delta/2}2^{\sqrt{\log n}(1+\Delta)\Delta/2}=Q(n)^{(\Delta+3)\Delta/2}n^{o(1)}.\qedhere
\]
\end{proof}
So for the ``fast query'' case data structure, by picking $Q(n)=\log^{O(1)}n$, 
we obtain a space lower bound of $S(n)=\domega(n^{\Delta+1})$.

\subsection{2D Polynomial Slab Stabbing}\label{subsec:polystab}

By small modifications,
our construction can also be applied to obtain a lower bound
for (the reporting version of) polynomial slab stabbing problems
using Theorem~\ref{thm:afshani-framework}.

One modification is that  we need to generate the slabs in such a way that they cover the entire
square $S$. 
The framework provided through Theorem~\ref{thm:afshani-framework} is more stream-lined and
derandomization is not needed and we can directly apply
the ``volume upper bound'' obtained through Lemma~\ref{lem:poly-int}.
There is also no $n^{o(1)}$ factor loss (our lower bound actually uses $\Omega(\cdot)$ notation).
The major change is that we need to use different parameters since we need to create $n$ polygonal slabs,
as now they are the input. 

\begin{theorem}
\label{thm:poly-slab-stab}
Give a set $\sR$ of $n$ 2D polynomial slabs $\{(P(x),w):deg(P)=\Delta \ge 2, w>0\}$, 
any data structure for $\sR$ solving the 2D polynomial slab stabbing problem with query time $Q(n)+O(k)$
uses $S(n)=\Omega\left(\frac{n^{1+2/(\Delta+1)}}{Q(n)^{2/\Delta}}\right)$ space,
where $k$ is the output size.
\end{theorem}

\begin{proof}
We use Afshani's lower bound framework as described in Theorem~\ref{thm:afshani-framework}.
First we generate $n$ polynomial slabs in a unit square $S=[0,1]\times[0,1]$ as follows.
Consider $\Theta\left(\frac{1}{2^{\Delta}\prod_{i=1}^{\Delta}d_i}\cdot\frac{\Delta}{w}\right)$ polynomial slabs 
with their base polynomials being:
\[
P(x)=\sum_{i=1}^\Delta j_id_ix^i+kw,
\]
where $j_i = \lfloor\frac{1}{2d_i}\rfloor, \lfloor\frac{1}{2d_i}\rfloor+1, \cdots, \lfloor\frac{1}{d_i}\rfloor$
for $1 \le i \le \Delta$
and $k = -\lceil\frac{\Delta}{w}\rceil, -\lceil\frac{\Delta}{w}\rceil+1, \cdots, \lceil\frac{\Delta}{w}\rceil$
for parameters $d_i$, $i=1,2,\cdots,\Delta$ and $w$ to be specified later.
Note that by our construction,
any point in the unit square is covered by $t=\Theta(1/(2^{\Delta}\prod_{i=1}^{\Delta}d_i))$ polynomial slabs.
To see this, consider each polynomial of form
\[
P'(x)=\sum_{i=1}^\Delta j_id_ix^i,
\]
where $j_i = \lfloor\frac{1}{2d_i}\rfloor, \lfloor\frac{1}{2d_i}\rceil+1, \cdots, \lfloor\frac{1}{d_i}\rfloor$
for $1 \le i \le \Delta$.
By shifting $P'(x)$ vertically with distance $kw$ for 
$k = -\lceil\frac{\Delta}{w}\rceil, -\lceil\frac{\Delta}{w}\rceil+1, \cdots, \lceil\frac{\Delta}{w}\rceil$,
we generate a series of adjacent 2D polynomial slabs.
For $x\in[0,1]$, $0 \le P'(x) \le \Delta$ for each $P'(x)$.
The maximum value is achieved by picking $j_i = \lfloor1/d_i\rfloor$ for each $i =1, 2, \cdots, \Delta$.
Observe that $S$ is completely contained between polynomials $y=P'(x)-\lceil\frac{\Delta}{w}\rceil w$ 
and $y=P'(x)+\lceil\frac{\Delta}{w}\rceil w$. 
As a result,  for each $P'(x)$ (i.e., each choice of indices $j_i$'s), the square $S$ is covered
exactly once. 
This implies $S$ is covered 
$t = \prod_{i=1}^{\Delta}\left(\lfloor\frac{1}{d_i}\rfloor - \lfloor\frac{1}{2d_i}\rfloor + 1 \right)=\Theta(1/(2^{\Delta}\prod_{i=1}^{\Delta}d_i))$ times,
which is the number of choices we have for the indices $j_1, \cdots, j_\Delta$.
Therefore, each point in $S$ is covered by $t$ polynomial slabs.

We set $d_i=c_1\left(\frac{1}{Q(n)^{2/(\Delta(\Delta+1))}w^{2/(\Delta+1)}}\right)^iw$ for some sufficiently small constant $c_1$
such that $S$ is covered by $t \ge Q(n)$ polynomial slabs.
We then set $w=c_2\frac{Q(n)}{n}$ for some suitable constant $c_2$ according to $c_1$ such that
the total number of polynomial slabs we have generated is exactly $n$.

By Lemma~\ref{lem:poly-int} and a similar argument as in the proof of Theorem~\ref{thm:general-ring-lb},
the intersection area of any two polynomial slabs in our construction
is bounded by
\[
v\le O(\Delta) w(\Delta+1)^3 \left(\frac{w}{d_i}\right)^{\frac{1}{i}}
=f(\Delta)\frac{Q(n)^{1 + \frac{2}{\Delta} }}{ n^{1 + \frac{2}{\Delta+1} }},
\]
where $f(\Delta)=O(1)$ is some value depending on $\Delta$ only
and $i=1,2,\cdots,\Delta$ is the largest degree where the two polynomials differ in their coefficients.
Since $t\ge Q(n)$, then according to Theorem~\ref{thm:afshani-framework},
\[
S(n)=\Omega\left(\frac{t}{v}\right)
=\Omega\left(\frac{n^{1 + \frac{2}{\Delta+1} }}{Q(n)^{ \frac{2}{\Delta} }}\right).
\qedhere
\]
\end{proof}

So for any data structure that solves the 2D polynomial slab stabbing problem
using $S(n)=O(n)$ space,
Theorem \ref{thm:poly-slab-stab} implies that its query time must be
$Q(n)=\Omega(n^{1-1/(\Delta+1)})$.

\section{2D Annulus Reporting and Stabbing}
\subsection{2D Annulus Reporting}
In this subsection, 
we show that any data structure that solves 2D annulus reporting with $\log^{O(1)}n$ query time
must use $\domega(n^3)$ space.
Recall that an annulus is the region between two concentric circles and the \textit{width} of the 
annulus is the difference between the radii of the two circles. 
In general, we show that if the query time is $Q(n) + O(k)$, then the data structure
must use $\domega(n^3/Q(n)^5)$ space.
We will still use Chazelle's framework.

We first present a technical geometric lemma which upper bounds the intersection area of two 2D annuli.
We will later use this lemma to show that with probability more than $1/2$, 
a random point set satisfies the first condition of Theorem~\ref{thm:chazelle-framework}.

\begin{restatable}{lemma}{generalringint}\label{lem:general-ring-int}
  Consider two annuli of width $w$ with inner radii of $r_1,r_2$,
where $r_1+w \le r_2, w < r_1$, and $r_1,r_2=\Theta(n)$.
Let $d$ be the distance between the centers of two annuli.
When $w \le d < r_2$,
the intersection area of two annuli is bounded by $O\left(wn\sqrt{\frac{w^2}{(g+w)d}}\right)$,
where $g=\max\{r_1-r_2+d,0\}$.
\end{restatable}

\begin{proof-sketch}
For the complete proof see Appendix~\ref{sec:proof-general-ring-int}.
When $w \le d \le r_2 - r_1 + 2w$, the intersection region consists of two triangle-like regions.
We only bound the triangle-like region $\tilde{\triangle} PQR$ in the upper half annuli as shown in Figure~\ref{fig:maxintring}.
We can show that its area is asymptotically upper bounded by the product of its base length $|QR|=O(w)$ and its height $h$.
We bound $h$ by observing that $\frac{hd}{2}$ is the area of triangle
$\triangle P O_1 O_2$ but we can also obtain its area of using Heron's formula, given its
three side lengths. 
This gives $h = O(n\sqrt{w/d})$.
Since in this case $g\le 2w$, the intersection area is upper bounded by $O\left(wn\sqrt{\frac{w^2}{(g+w)d}}\right)$ as claimed.
\begin{figure}[H]
  \centering
  \includegraphics[scale=0.5]{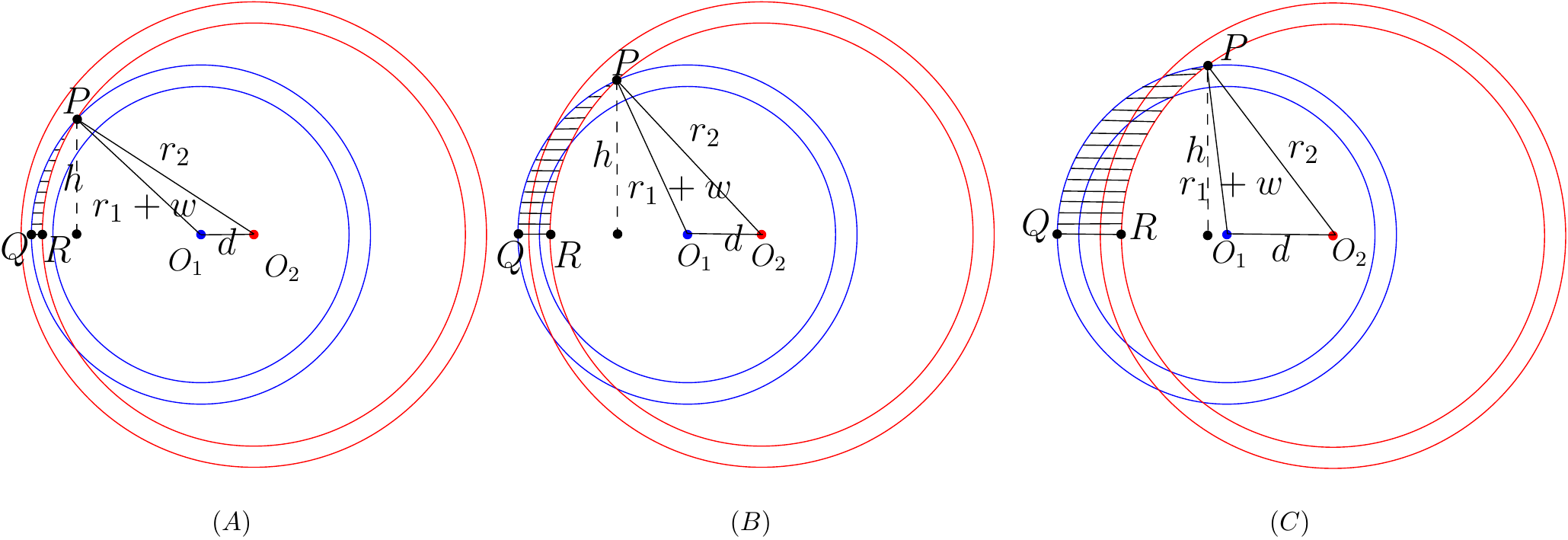}
  \caption{Intersections When $d$ is Small}
  \label{fig:maxintring}
\end{figure}
When $r_2-r_1+2w \le d \le r_2$,
the intersection region consists of two quadrilateral-like regions.
Again we only consider $\tilde{\square} ABCD$ in the upper half of the annuli,
which is contained in a partial annulus, $\tilde{\rR}_{EFHD}$, as shown in Figure~\ref{fig:mid-int-vol}.
\begin{figure}[H]
     \centering
     \begin{minipage}[b]{.4\textwidth}
         \centering
         \includegraphics[width=0.7\textwidth]{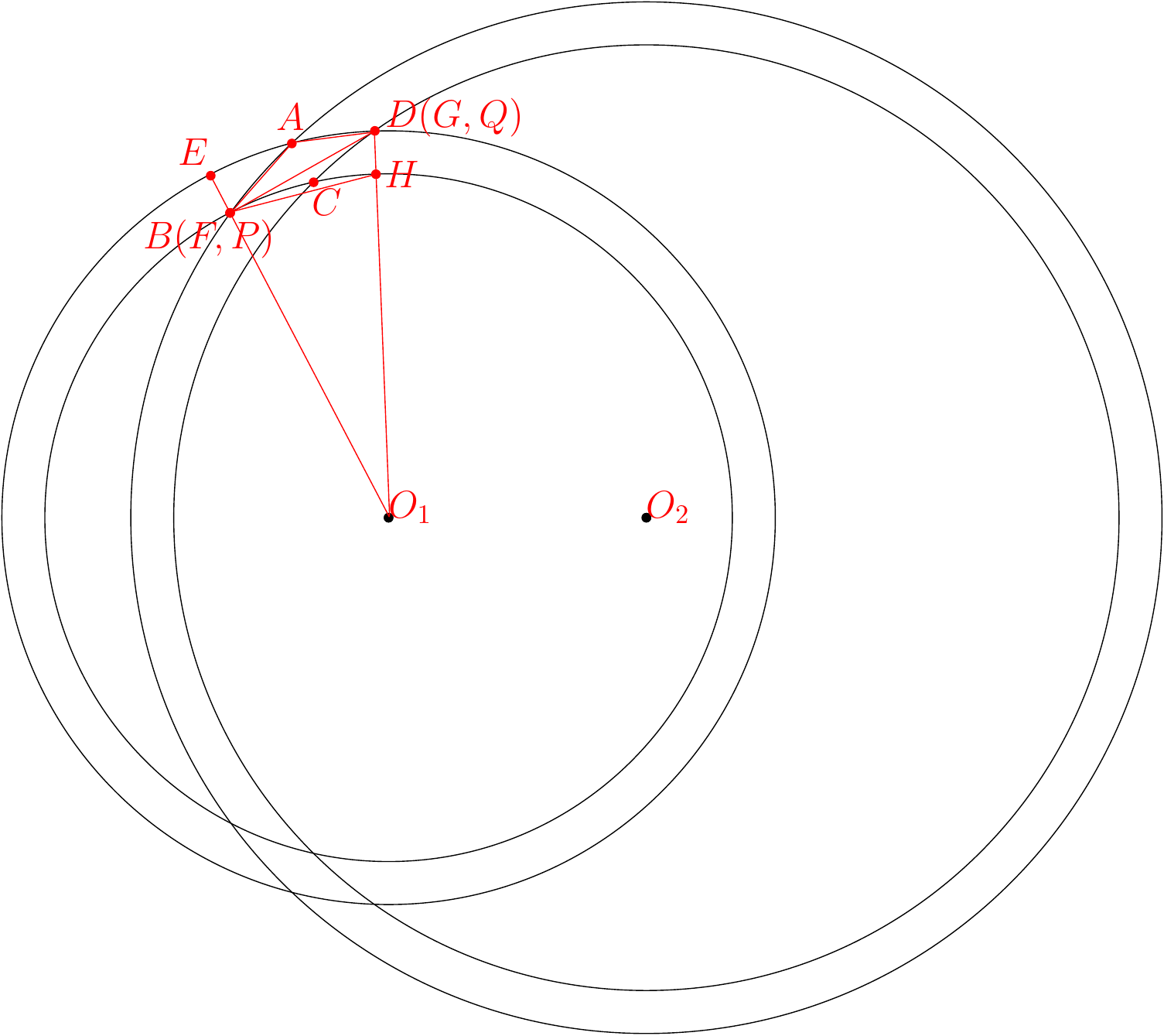}
         \subcaption{Cover an Intersection by A Partial Annulus}
         \label{fig:mid-int-vol}
     \end{minipage}
     \hfill
     \begin{minipage}[b]{.4\textwidth}
         \centering
         \includegraphics[width=0.7\textwidth]{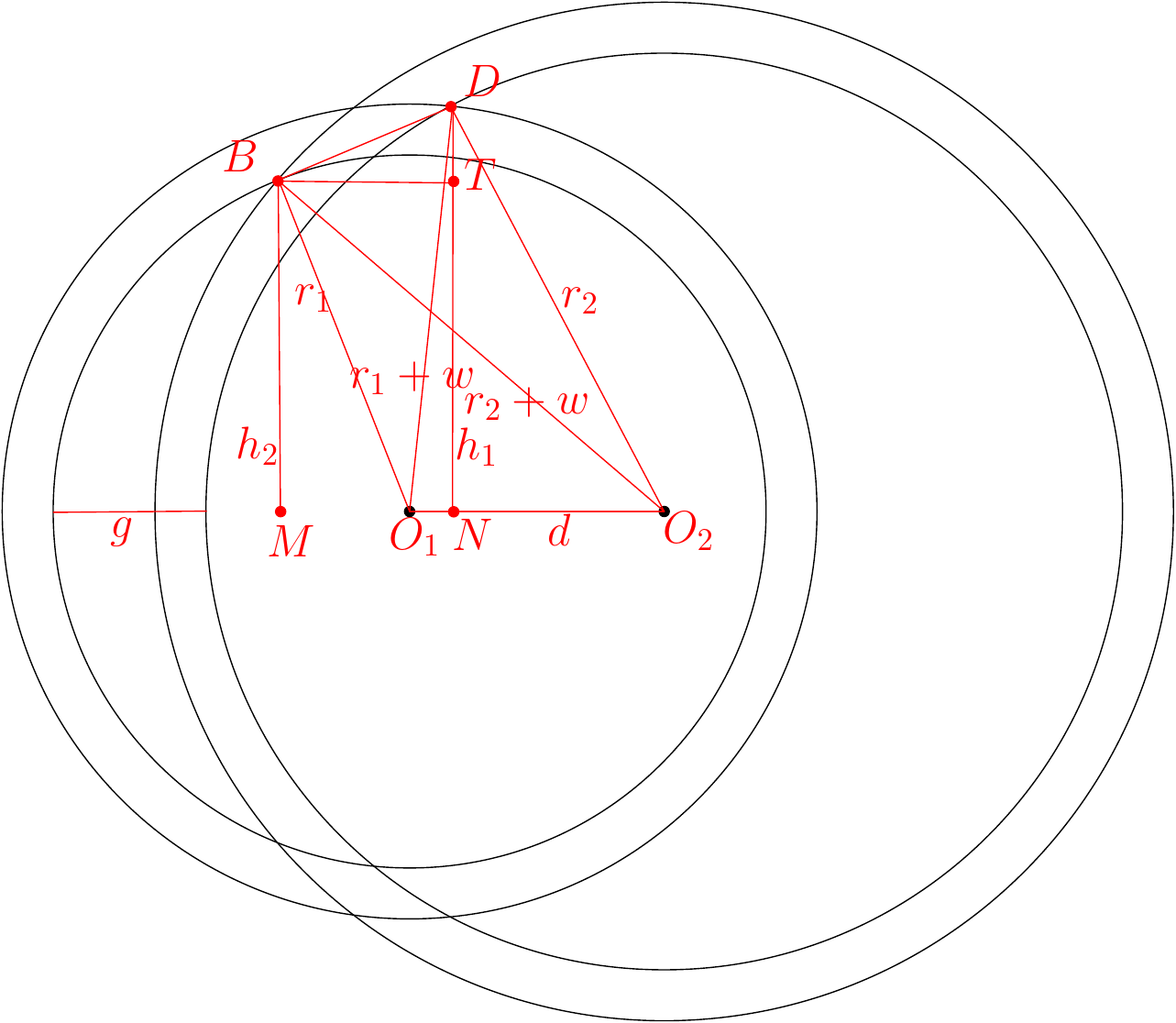}
         \subcaption{Bound the Length of $|BD|$}
         \label{fig:bound-bd}
     \end{minipage}
\caption{Cover a Quadrilateral-like Region by a Partial Annulus}
\label{fig:mid-int}
\end{figure}
We show the area of $\tilde{\rR}_{EFHD}$ is asymptotically bounded by $|BH|\cdot w$,
where $|BH|$ is the distance between the two endpoints of the inner arc.
We upper bound $|BH|$ by $|BD|$.
We use the algebraic representation of the two annuli,
to bound the length of the projection of $BD$ on the $X$-axis
by $\Theta\left(\frac{wn}{d}\right)$; See Figure~\ref{fig:bound-bd}.
We use Heron's formula to bound the length of the projection of $BD$ on the $Y$-axis
by $O\left(n\sqrt{w^2/dg}\right)$.
The maximum of the length of the two projections yields the claimed bound.
\end{proof-sketch}

We use Chazelle's framework to obtain a lower bound for 2D annulus reporting.
Let $S_1$ and $S_2$ be two squares of side length $n$ that are placed $10n$ distance apart and
$S_2$ is directly to the left of $S_1$. 
We generate the annuli  as follows. 
We divide $S_1$ into a $\frac{n}{T}\times\frac{n}{T}$ grid 
where each cell is a square of side length $T$.
For each grid point, we construct a series of circles as follows.
Let $O$ be a grid point.
The first circle generated for $O$ must pass through a corner of $S_2$ and not intersect 
the right side of $S_2$, as shown in Figure~\ref{fig:2d-ring-family}.
Then we create a series of circles centered at $O$ by increasing the radius by increments of $w$, for some $w<T$,
as long as it does not  intersect the left side of $S_2$.
Every consecutive two circles defines an annulus centered on $O$.
We repeat this for every grid cell in $S_1$ and this makes up our set of queries.
The input points are placed uniformly randomly inside $S_2$.
\begin{figure}[H]
  \centering
  \includegraphics[scale=0.7]{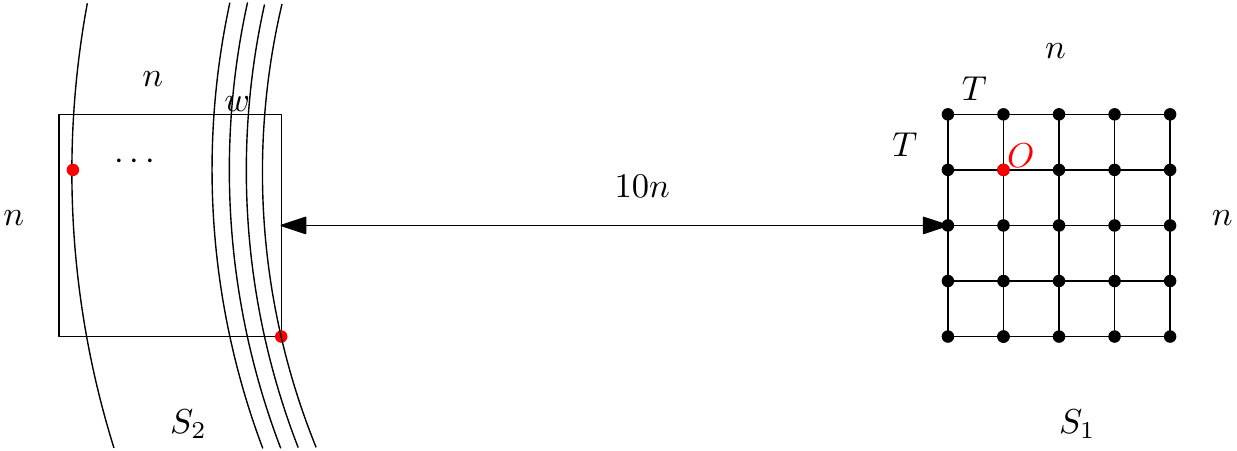}
  \caption{Generate a Family of Annuli at Point $O$}
  \label{fig:2d-ring-family}
\end{figure}

We now show that for the annuli we constructed,
the intersection of $\ell$ annuli is not too large,
for some $\ell$ we specify later.
More precisely we prove the following. 

\begin{restatable}{lemma}{numintring}
\label{lem:num-int-ring}
    There exists a large enough constant $c$ such that 
    in any subset of $\ell=cw^2/\sqrt{T}$ annuli, we can find two annuli such 
    that their intersection has area 
    $O\left(nw\sqrt{\frac{1}{T}}\right)$.
\end{restatable}

\begin{proof-sketch}
For the complete proof see Appendix~\ref{sec:proof-num-int-ring}.
Let $\sS$ be a set of $\ell=cw^2/\sqrt{T}$ annuli.
Suppose for the sake of contradiction that we cannot find two annuli in $\sS$
whose intersection area is $O\left(nw\sqrt{\frac{1}{T}}\right)$.
Since by Lemma~\ref{lem:general-ring-int},
the intersection area of any two annuli in our construction 
with distance $\Omega(wT)$ is $O\left(nw\sqrt{\frac{1}{T}}\right)$.
The maximum distance between any two annuli in $\sS$ must be $o(wT)$.

\begin{figure}[H]
  \centering
  \includegraphics[scale=0.4]{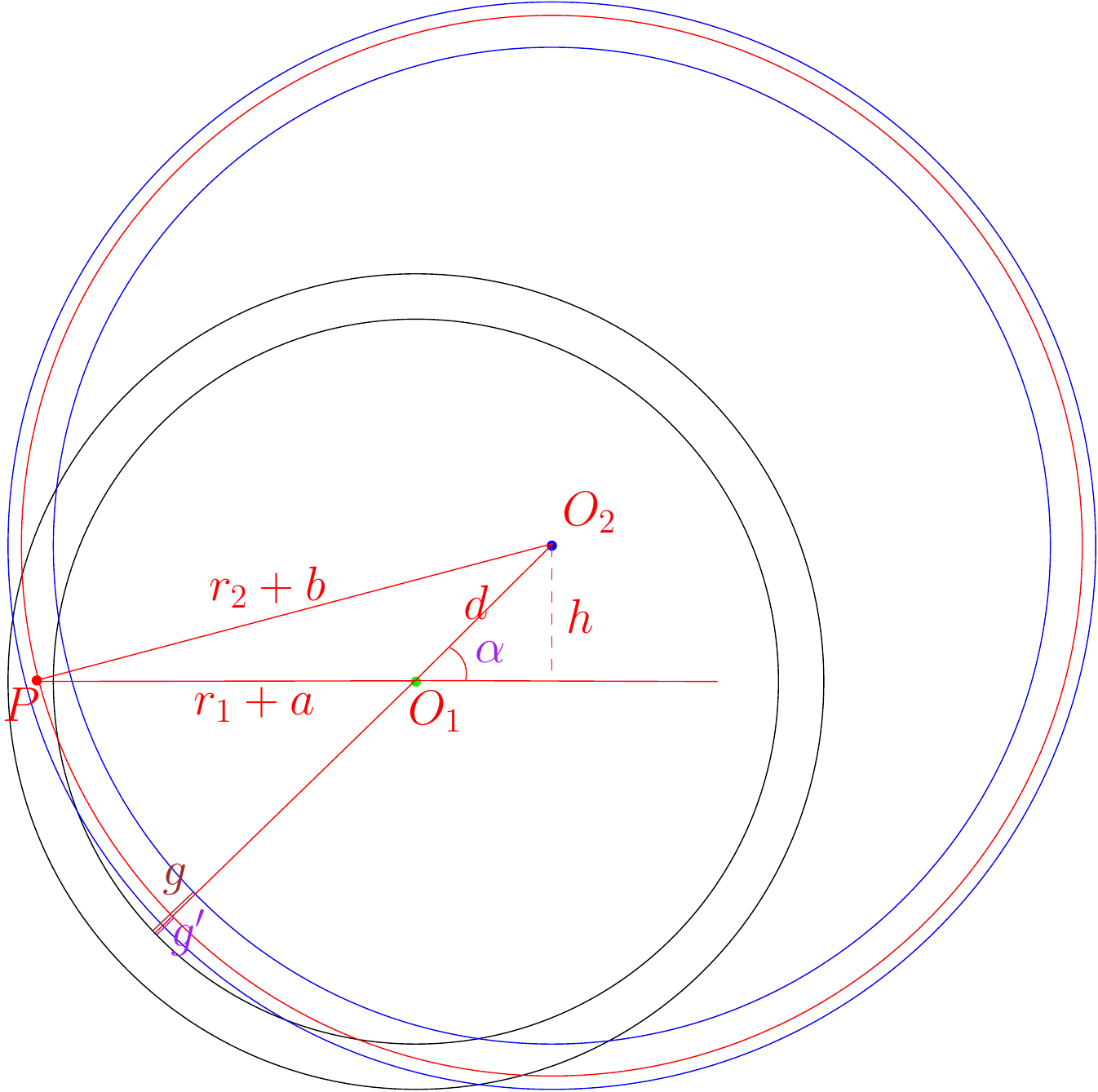}
  \caption{Intersection of Two Annuli}
  \label{fig:bound-region}
\end{figure}

Let $P$ be a point in the intersection of annuli in $\sS$.
Consider an arbitrary annulus $\rR_1 \in \sS$ centered at $O_1$
and another annulus $\rR_2 \in \sS$ centered at $O_2$ for some $O_2 \notin PO_1$.
For $\rR_1, \rR_2$ to contain $P$, we must have $|PO_1|=r_1+a, |PO_2|=r_2+b$ for $0 \le a, b \le w$.
See Figure~\ref{fig:bound-region} for an example.
Also $|O_1O_2|=d$,
by exploiting the shape of $\triangle PO_1O_2$
and applying Lemma~\ref{lem:general-ring-int},
we can compute an upper bound for the distance between $O_2$ and $PO_1$,
namely, $h=d\sin\alpha=o(w\sqrt{T})$,
where $\alpha$ is the angle between $O_1O_2$ and $PO_1$.
This implies that $\sS$ must fit in a rectangle of size $o(wT)\times o(w\sqrt{T})$.
Since the gird cell size is $T\times T$,
only $o(w^2/\sqrt{T})$ annuli are contained in such a rectangle,
a contradiction.
\end{proof-sketch}

We are now ready to plug in some parameters in our construction. 
We set $T=w^22^{2\sqrt{\log n}}$.
First, we claim  that from each grid cell $O$, we can draw $\Theta(n/w)$ circles;
Let $C_1, C_2, C_3$, and $C_4$ be the corners of $S_2$ sorted increasingly according to their
distance to $O$.
As $S_1$ and $S_2$ are placed $10n$ distance apart, an elementary geometric calculation reveals that
$C_1$ and $C_2$ are vertices of the right edge of $S_2$, meaning, the smallest circle that we draw
from $O$ passes through $C_2$ and we keep drawing circles, by incrementing their radii by $w$
until we are about to draw a circle that is about to contain $C_3$.
We can see that $|OC_3| - |OC_2| = \Theta(n)$ and thus we draw
$\Theta(n/w)$ circles from $O$. 
As we have $\Theta((n/T)^2)$ grid cells, it thus follows that we have
$\Theta(n^3 / (T^2w))$ annuli in our construction. 

Also by our construction, the area of each annulus within $S_2$ is 
$\Theta(wn)$.
To see this,
let $P$ be an arbitrary point in $S_1$,
let $A, B$ be the intersections of some circle centered at $P$ as in Figure~\ref{fig:2d-ring-angle}.
\begin{figure}[H]
  \centering
  \includegraphics[scale=0.8]{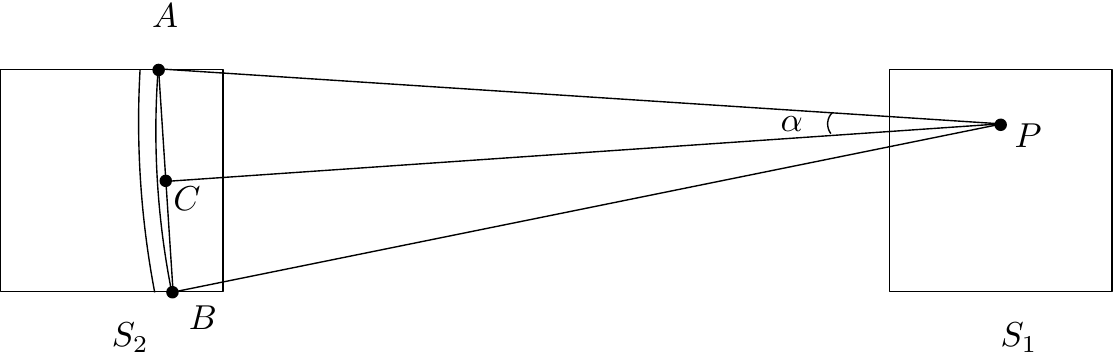}
  \caption{The Angle of an Annulus}
  \label{fig:2d-ring-angle}
\end{figure}
We connect $AB$ and let $C$ be the center of $AB$.
Let $\alpha=\angle{APC}$.
In the triangle $\triangle ABP$, all the sides are within constant factors of each other
and thus $\alpha=\Theta(1)$ and so the area of the annulus inside $S_2$ is at least a constant
fraction of the area of the entire annulus. 

Suppose we have a data structure that answers 2D annulus reporting queries in $Q(n)+O(k)$ time.
We set $w=c'Q(n)$ for a large enough constant $c'$ such that 
the area of each annulus within $S_2$ is at least $\Theta(w n)>8nQ(n)$.
By Lemma~\ref{lem:derand-ring}, if we sample $n$ points uniformly at random in $S_2$,
then with probability more than $1/2$, each annulus contains at least $Q(n)$ points.

Also by our construction, 
the total number of intersections of two annuli is bounded by $O(n^6)$
and by our choice of $T$, $O\left(nw\sqrt\frac{1}{T}\right)=O\left(\frac{n}{2^{\sqrt{\log n}}}\right)$.
Then by Lemma~\ref{lem:derand-int} and Lemma~\ref{lem:num-int-ring}, with probability $>\frac{1}{2}$,
a point set of size $n$ picked uniformly at random in $S_2$
satisfies that the number of points in any of the intersection of $cw^2/\sqrt{T}$ annuli
is no more than $9\sqrt{\log n}$.

Now by union bound,
there exist $\Theta\left(\frac{n^3}{wT^2}\right)$ point sets
such that each set is the output of some 2D annulus query
and each set contains at least $Q(n)$ points.
Furthermore, the intersection of any $cw^2/\sqrt{T}$ sets is bounded by $9\sqrt{\log n}$.
Then by Theorem~\ref{thm:chazelle-framework}, we obtain a lower bound of
\[
S(n)=\Omega\left(\frac{Q(n)n^3\sqrt{T}}{wT^2w^22^{O\left(\sqrt{\log n}\right)}}\right)
=\domega\left(\frac{n^3}{Q(n)^5}\right).
\]

This proves the following theorem about 2D annulus reporting.

\begin{theorem}
\label{thm:2d-ring-lb}
Any data structure that solves 2D annulus reporting on point set of size $n$ with query time $Q(n)+O(k)$,
where $k$ is the output size,
must use $\domega\left({n^3}/{Q(n)^5}\right)$ space.
\end{theorem}

So for any data structure that solves 2D annulus reporting in time $Q(n)=\log^{O(1)}n$,
Theorem~\ref{thm:2d-ring-lb} implies that $\domega\left(n^{3}\right)$ space must be used.

\subsection{2D Annulus Stabbing}
Modifications similar to those done in Subsection~\ref{subsec:polystab}
can be used to obtain the following lower bound. 

\begin{theorem}
\label{thm:ring-stab}
Any data structure that solves 
the 2D annulus stabbing problem with query time $Q(n)+O(k)$, where $k$ is the output size,
must use $S(n)=\Omega(n^{3/2}/Q(n)^{3/4})$ space.
\end{theorem}

\begin{proof}
We use Afshani's lower bound framework as in Theorem~\ref{thm:afshani-framework}.
We construct annuli similar to the way as we did in the proof of Theorem~\ref{thm:2d-ring-lb},
but with some differences.
Let $S_1$ be a unit square,
we decompose $S_1$ into a grid where each grid cell is a square of size $T\times T$,
for some parameter $T$ to be determined later.
Let $S_2$ be another unit square to the left of $S_1$ with distance $10$.
For each grid point in $S_1$,
we generate a family of circles where the first circle is the first one
tangent to the right side of $S_2$.
See Figure~\ref{fig:2d-ring-stabbing} for an example.
\begin{figure}[h]
  \centering
  \includegraphics[scale=0.8]{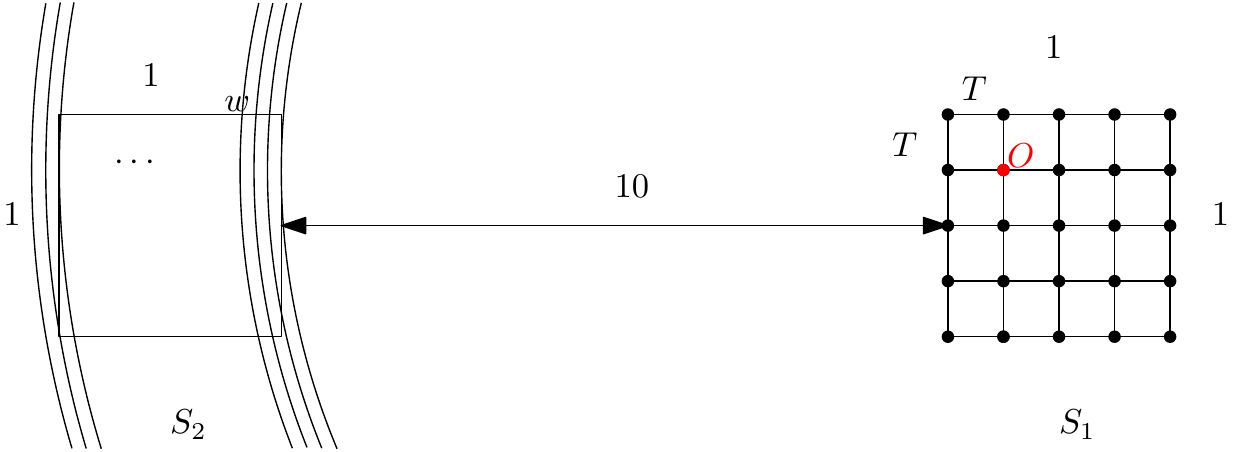}
  \caption{Generate a Family of Annuli at Point $O$}
  \label{fig:2d-ring-stabbing}
\end{figure}
Let $r_0$ be the radius of this circle.
We generate other circles by increasing the radius by $w$ each time.
The last circle is the one with radius at least $r_0+(\sqrt{122}-9)$.
The choice of constant $\sqrt{122}-9$ is due to the following.
Consider a point $P$ in the upper half of $S_1$,
let $a$ (resp. $b$) be the distance from $P$ to the upper (resp. left) side of $S_1$.
The case for the lower half of $S_1$ is symmetric.
Consider the last circle intersecting $S_2$ only at the corners.
See Figure~\ref{fig:2d-ring-stabbing-max-r} for an example.
\begin{figure}
  \centering
  \includegraphics[scale=0.8]{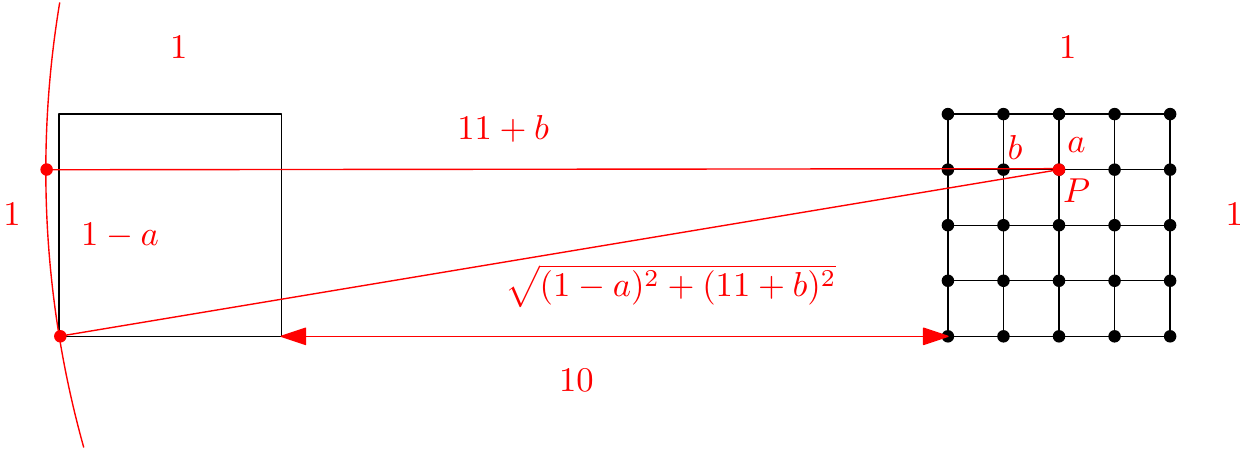}
  \caption{The Last Valid Circle}
  \label{fig:2d-ring-stabbing-max-r}
\end{figure}
The difference of radius lengths between the last circle and the first circle is
\[
f=\sqrt{(1-a)^2+(11+b)^2}-10-b.
\]
Simple analysis shows that $f\le \sqrt{122}-10$.
This shows as long as the last circle has radius at least $\sqrt{122}-9+r_0$,
it is completely outside of $S_2$.
Like we did previously,
we consider the region between two consecutive annuli to be an annulus.
Note that by our construction,
for any point in $S_2$, it is contained in one of the annuli we generated 
for a given grid point.
The total number of annuli we have generated is therefore
\[
\left(\frac{1}{T}+1\right)^2\cdot\frac{\sqrt{122}-9}{w}.
\]

We set $T=1/(2\sqrt{Q(n)}-1)$ and $w=4(\sqrt{122}-9)Q(n)/n$.
Then the total number of annuli we have generated is $n$.
Furthermore, each query point is contained in $t=(1/T+1)^2 \ge Q(n)$ annuli.

To use Lemma~\ref{lem:general-ring-int}, 
we require $w\le T$, which implies $Q(n)=O(n^{2/3})$,
the intersection area of two annuli in our construction
is upper bounded by $v=O(w\sqrt{\frac{w}{T}})=O(Q(n)^{7/4}/n^{3/2})$.
Then by Theorem~\ref{thm:afshani-framework},
\[
S(n)=\Omega\left(\frac{t}{v}\right)=\Omega\left(\frac{n^{3/2}}{Q(n)^{3/4}}\right).
\qedhere
\]
\end{proof}

So for any data structure that solves the 2D annulus stabbing problem
using $O(n)$ space, Theorem~\ref{thm:ring-stab} implies that
its query time must be $Q(n)=\Omega(n^{2/3})$.

\section{Conclusion and Open Problems}

We investigated lower bounds for range searching
with polynomial slabs and annuli in $\bR^2$.
We showed space-time tradeoff bounds of  $S(n)=\domega(n^{\Delta+1}/Q(n)^{(\Delta+3)\Delta/2})$
and $S(n)=\domega(n^{3}/Q(n)^5)$ for them respectively.
Both of these bounds are almost tight in the ``fast query'' case, i.e., when $Q(n)=\log^{O(1)}n$
(up to a $n^{o(1)}$ factor).
This refutes the conjecture of the existence of data structure that can solve
semialgebraic range searching in $\bR^d$ using $\dO(n^d)$ space and $\log^{O(1)}n$ query time.
We also studied the ``dual'' polynomial slab stabbing and annulus stabbing problems.
For these two problems, we obtained lower bounds $S(n)=\Omega(n^{1+2/(\Delta+1)}/Q(n)^{2/\Delta})$
and $S(n)=\Omega(n^{3/2}/Q(n)^{3/4})$ respectively.
These bounds are tight when $S(n)=O(n)$. 
Our work, however, brings out some very interesting open problems. 

To get the lower bounds for the polynomial slabs, 
we only considered univariate polynomials of degree $\Delta$.
In this setting, the number of coefficients is at most $\Delta+1$,
and we have also assumed they are all independent.
It would be interesting to see if similar lower bounds can be obtained
under more general settings. 
In particular, as the maximum number of coefficients of a bivaraite polynomial
of degree $\Delta$ is $\binom{\Delta+2}{2}$, it would interesting to see if
a $\domega(n^{ \binom{\Delta+2}{2}-1})$ space lower bound can be obtained
for the ``fast query'' case. 

It would also be interesting to consider space-time trade-offs. 
For instance, by combining the known ``fast query'' and ``low space'' solutions for  2D annulus
reporting, one can obtain data structures with trade-off curve $S(n) = \tilde{O}(n^3/Q(n)^4)$,
however, our lower bound is $S(n) = \domega(n^3/Q(n)^5)$ and it is not clear which of these
bounds is closer to truth.
For the annulus searching problem in $\bR^2$,
in our lower bound proof, we considered a random input point set,
since in most cases a random point set is the hardest input instance
and our analysis seems to be tight,
we therefore conjecture that our lower bound could be tight, at least
when $Q(n)$  is small enough.
We believe that it should be possible to obtain the trade-off curve
of $S(n) = \tilde{O}(n^3/Q(n)^5)$ when the input
points are uniformly random in the unit square
and $Q(n)$ is not too big.

Finally, another interesting direction is to study the lower bound for the counting variant
of semialgebraic range searching.

\begin{acks}
The authors would like to thank Esther Ezra for sparking the initial ideas behind the proof.
\end{acks}

\bibliography{reference}{}
\bibliographystyle{ACM-Reference-Format}

\appendix

\section{Proof of Lemma~\ref{lem:general-ring-int}}
\label{sec:proof-general-ring-int}
\generalringint*

\begin{proof}
First note that for two annuli to intersect,
$d\ge r_2-r_1-w$.
When $w\le d\le r_2-r_1+2w$,
observe that $g\le 2w$ and
the intersections are all bounded by a two triangle-like regions.
We consider the triangle like region $\tilde\triangle PQR$
in the upper half of the annuli
as shown in Figure~\ref{fig:ap-maxintring}.
The area of the other triangle like region
can be bounded symmetrically.

\begin{figure}[h]
  \centering
  \includegraphics[scale=0.5]{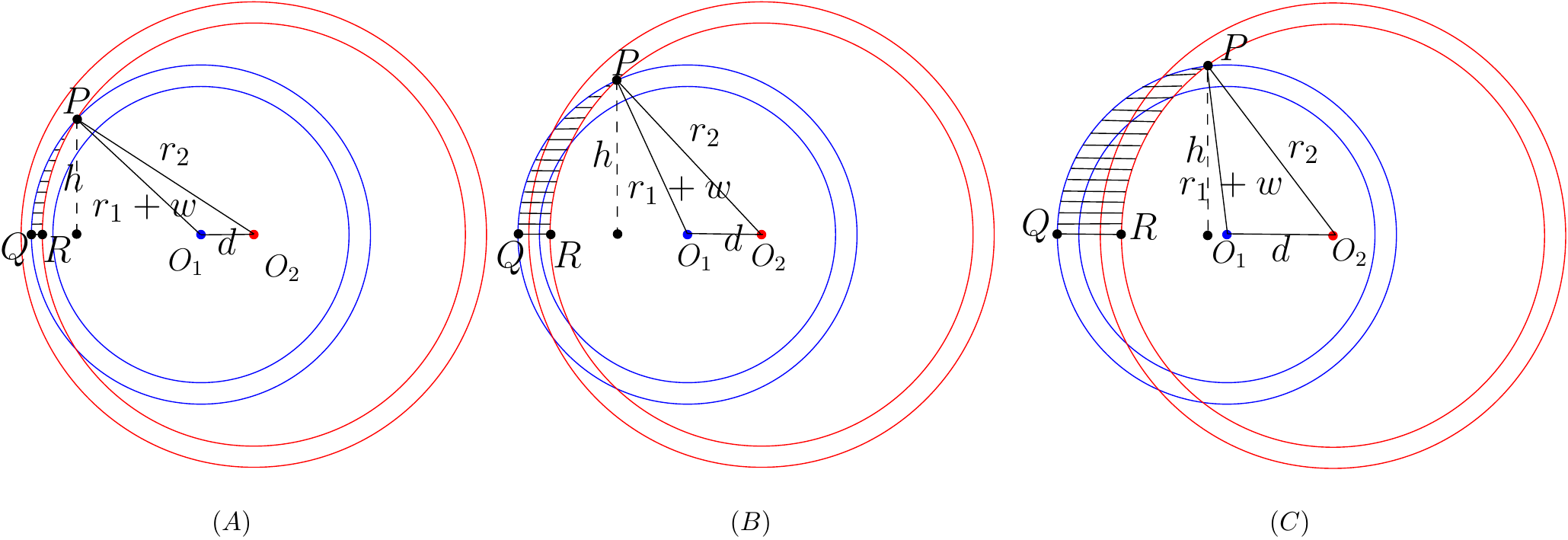}
  \caption{Intersections When $d$ is Small}
  \label{fig:ap-maxintring}
\end{figure}

To upper bound its area, we first compute its ``height'' $h$.
Consider $\triangle PO_1O_2$, by Heron's Formula, its area is

\begin{align*}
A_{\triangle PO_1O_2}
&=\sqrt{\frac{r_1+r_2+w+d}{2}\cdot\frac{r_2+d-r_1-w}{2}\cdot\frac{r_1+d+w-r_2}{2}\cdot\frac{r_1+r_2+w-d}{2}}\\
&=O(\sqrt{r_2\cdot d\cdot w\cdot r_1})\\
&=O\left(n\sqrt{dw}\right),
\end{align*}
where the second equality follows from
\begin{align*}
d\le r_2-r_1+2w~&\implies~ r_1+r_2+w+d\le 2r_2+3w\le 5r_2,\\
r_2-r_1-w\le d~&\implies~r_2+d-r_1-w\le 2d,\\
d\le r_2-r_1+2w~&\implies~r_1+d-r_2+w\le 3w,\\
r_2-r_1-w\le d~&\implies~r_1+r_2+w-d\le2r_1+2w\le4r_1.
\end{align*}

Then we bound $h$ by
\[
h=\frac{2A_{\triangle PO_1O_2}}{d}=O\left(n\sqrt\frac{w}{d}\right).
\]
To bound the area of $\tilde\triangle PQR$,
we move the outer circle of the annulus centered at $O_1$ along $O_1O_2$ such that it passes $R$
as in Figure~\ref{fig:area-bound-init}.
Let $S$ be a point in this new circle such that $PS\parallel QR$.
The area of region $PSRQ$ formed by the outer circle of the annulus centered at $O_1$
and the new circle as well as $PS$ and $QR$, i.e., the shaded region in Figure~\ref{fig:area-bound-init},
is $O(wh)$ by a simple integral argument.
Since the curvature of the inner circle of the annulus centered at $O_2$ is no larger than that of
the outer circle of $O_1$ since $r_2\ge r_1+w$, $\tilde\triangle PQR$ is contained in this region.
So $A_{\tilde\triangle PQR}=O(wh)=O(nw\sqrt{w/d})=O(nw\sqrt{w^2/(g+w)d})$.

\begin{figure}[h]
  \centering
  \includegraphics[scale=0.5]{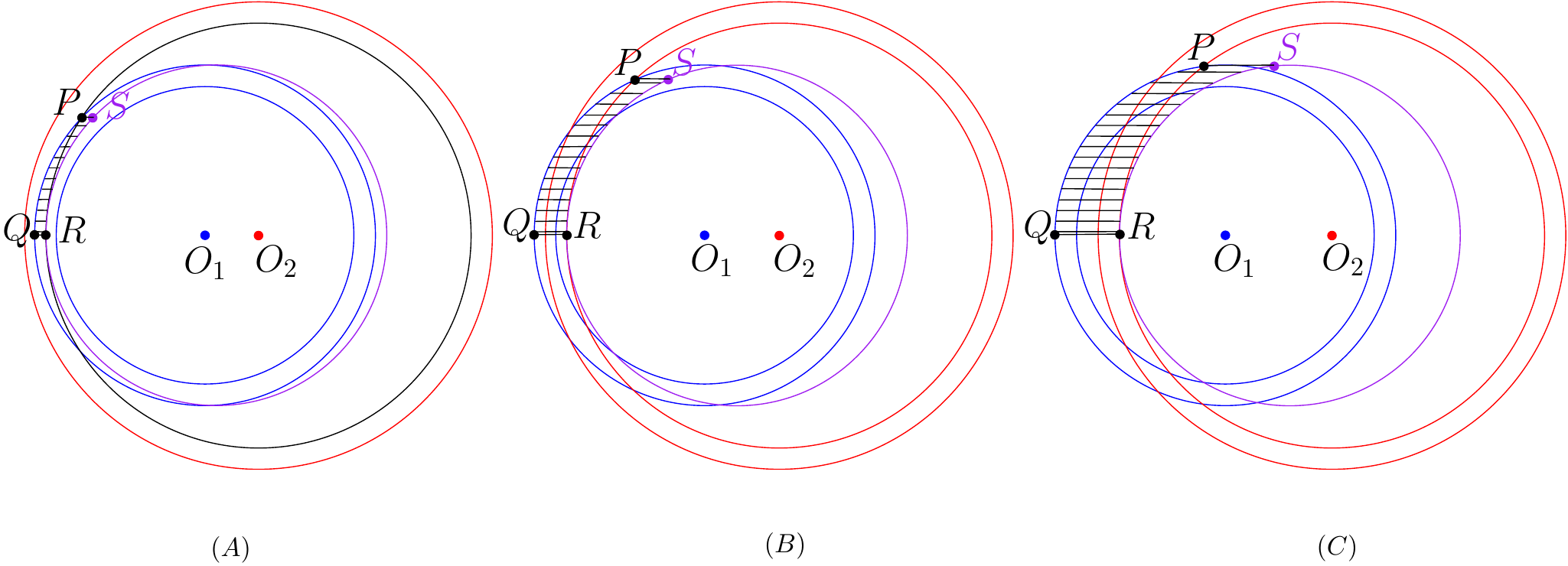}
  \caption{Bound the Intersection Area for Small $d$}
  \label{fig:area-bound-init}
\end{figure}

In the case $r_2-r_1+2w<d<r_2$, the intersections consist of two quadrilateral-like regions,
one at the upper half and the other at the lower half.
In this case, $g>2w$.
We show how to bound the area of the quadrilateral-like region at the upper half,
the other one can be bounded symmetrically.

First note that each quadrilateral-like region $\tilde\square{ABCD}$ is contained in a partial annulus as in Figure~\ref{fig:ap-mid-int}.
We generate this partial annulus by shooting a ray from the center $O_1$
and connect a point $P$ in $\arc{AB}$, such that $\tilde\square ABCD$ is completely to the right of this ray,
to create the left boundary $EF$,
where $E$ is a point in the outer circle of the annulus $\rR_1$ centered at $O_1$
and $F$ is a point in the inner circle of annulus $\rR_1$.
And similarly, shoot another ray and connect a point $Q$ in $\arc{CD}$ 
to create the right boundary $GH$.
Again, $G$ is a point in the outer circle of annulus $\rR_1$
and $H$ is a point in the inner circle of annulus $\rR_1$.
Note that since $d<r_2$,
by elementary geometry,
$P$ can only be $B$ and $Q$ can only be $D$ as in Figure~\ref{fig:ap-mid-int-vol}.

\begin{figure}[h]
     \centering
     \begin{minipage}[b]{.45\textwidth}
         \centering
         \includegraphics[width=.7\textwidth]{mid-int-vol}
         \subcaption{Cover an Intersection by A Partial Annulus}
         \label{fig:ap-mid-int-vol}
     \end{minipage}
     \hfill
     \begin{minipage}[b]{.4\textwidth}
         \centering
         \includegraphics[width=.7\textwidth]{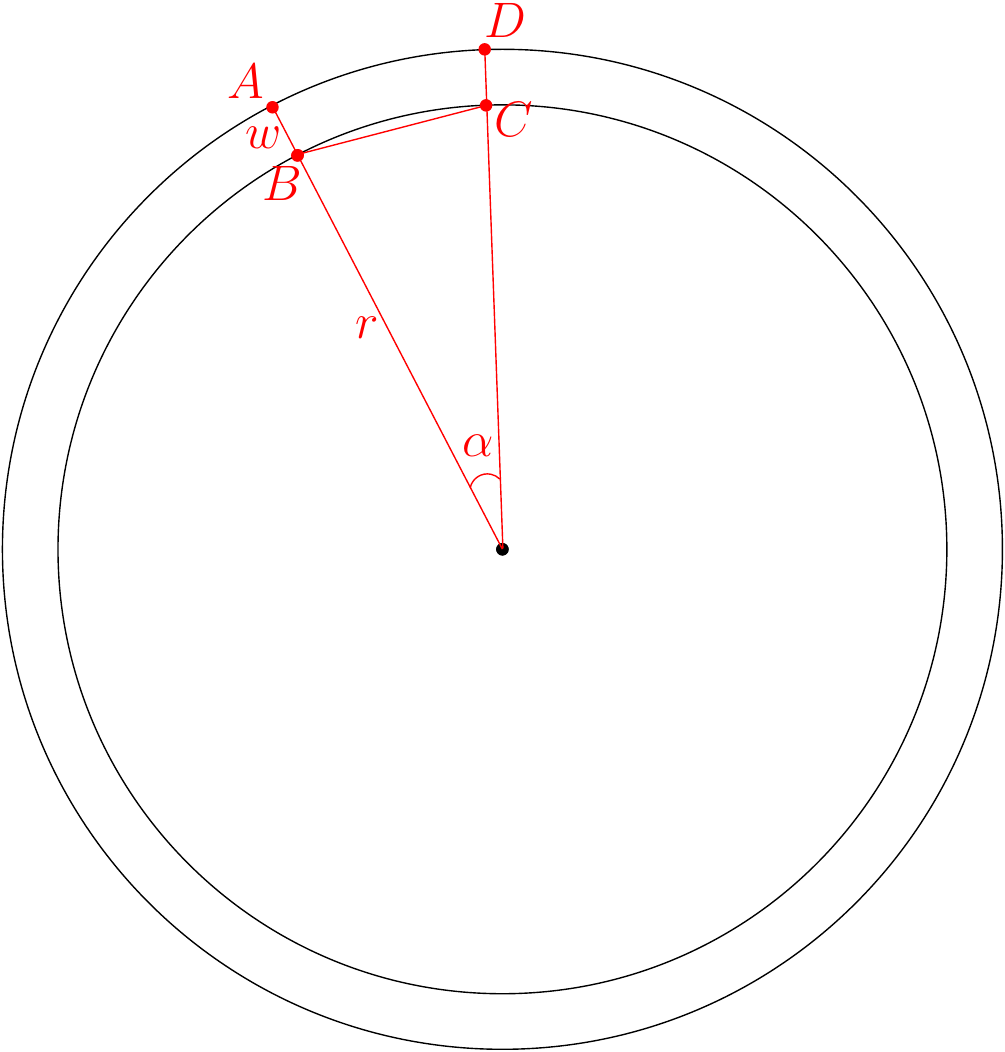}
         \subcaption{A Partial Annulus Example}
         \label{fig:ap-partial-ring}
     \end{minipage}
\caption{Cover a Quadrilateral-like Region by a Partial Annulus}
\label{fig:ap-mid-int}
\end{figure}

First, we observe that the area of the partial annulus is bounded by the product of 
the width of the annulus, $w$, and the length $FH$, within some constant factor.

To see this, consider the possible partial annulus generated by any two annuli,
see Figure~\ref{fig:ap-partial-ring} for an example.
One important observation is that for any intersection, the angle $\alpha$
of the partial annulus containing it is no more than $\pi$.
The area of the partial annulus is clearly $\alpha(rw+w^2/2)=\Theta(rw\alpha)$.
Note that for $0\le\alpha\le\pi$,
let $\beta=\alpha/2$, it is a simple fact that 
\[
\frac{1}{2}\beta\le\sin\beta\le\beta
\]
for $0\le\beta\le\pi/2$.
So for $\beta$ in this range, $\sin\beta=\Theta(\beta)$.
Thus we can compute $|FH|=2r\sin\beta=\Theta(r\alpha)$,
which implies that indeed $w\cdot|FH|$ gives us the area of the partial annulus,
within some constant factor.

Now we proceed to bound the length of $|FH|$.
Note that $|FH|\le|BD|$ because $\angle{BHD}>\frac{\pi}{2}$.
By Lemma~\ref{lem:diag-bound}, we know $|BD|=O\left(n\sqrt{\frac{w^2}{dg}}\right)$.
Therefore, the area of intersection is bounded by $O\left(wn\sqrt{\frac{w^2}{(g+w)d}}\right)$
as claimed.
\end{proof}

\begin{lemma}
\label{lem:diag-bound}
In $\mathbb{R}^2$, given two annuli centered at $O_1, O_2$ of width $w$ with inner radius lengths being $r_1, r_2$ respectively,
where $r_1+w\le r_2$, $w<r_1$, and $r_1,r_2=\Theta(n)$.
Let $d$ be the distance between two centers of the annuli and $g$ be the distance between two inner circles
of two annuli along the direction of $O_1O_2$.
Consider the quadrilateral-like region $\tilde\square ABCD$ formed by four arcs when $r_2-r_1+2w < d < r_2$,
we have $|BD|=O\left(n\sqrt{\frac{w^2}{gd}}\right)$
\end{lemma}
\begin{proof}
First note that $g$ is the distance between the two inner circles along $O_1O_2$,
as shown in Figure~\ref{fig:ap-bound-bd}.

\begin{figure}[h]
  \centering
  \includegraphics[scale=0.45]{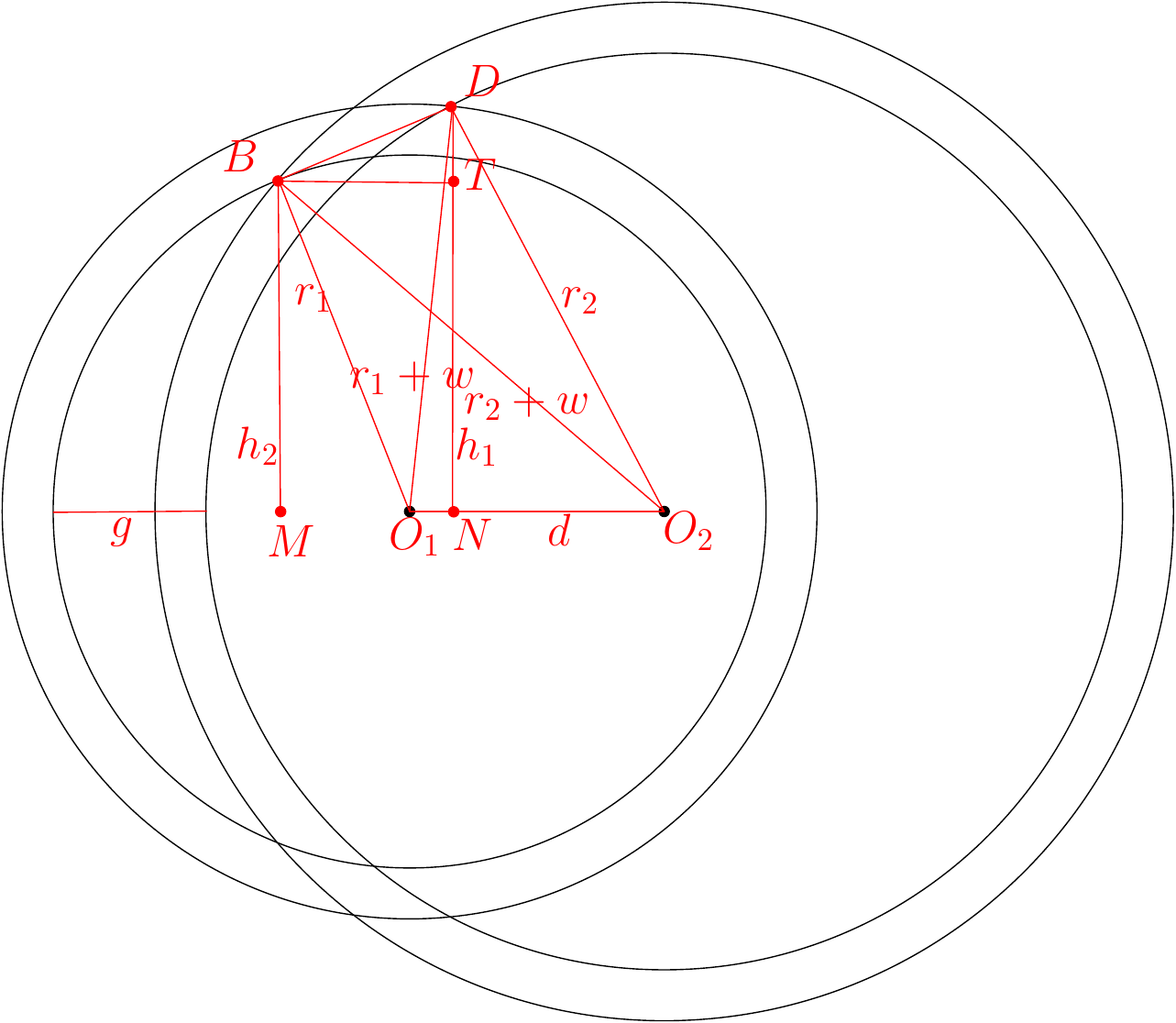}
  \caption{Bound the Length of $|BD|$}
  \label{fig:ap-bound-bd}
\end{figure}

We first compute $|BT|=x_{BD}=x_D-x_B$.
For point $B=(x_B,y_D)$, it satisfies
\[
\begin{cases}
x^2+y^2=r_1^2\\
(x-d)^2+y^2=(r_2+w)^2.
\end{cases}
\]
So $x_B$ satisfies
\[
2dx-d^2=r_1^2-(r_2+w)^2,
\]
which implies
\begin{align}
x_B=\frac{r_1^2-(r_2+w)^2+d^2}{2d}\label{eq:1}.
\end{align}
Similarly, we obtain
\begin{align}
x_D=\frac{(r_1+w)^2-r_2^2+d^2}{2d}\label{eq:2}.
\end{align}

So by (\ref{eq:1}) and (\ref{eq:2}), we obtain
\begin{align}
\label{eq:bt}
|BT|=x_{BD}=x_D-x_B=\frac{w(r_1+r_2+w)}{d}=\Theta\left(\frac{wn}{d}\right).
\end{align}

Now we compute $|DT|=y_{BD}=h_1-h_2$.
Let $\mu_1$ be the area of $\triangle{DO_1O_2}$
and $\mu_2$ be the area of $\triangle{BO_1O_2}$.

By Heron's formula,
\[
\begin{cases}
\mu_1=\sqrt{\frac{r_1+w+r_2+d}{2}\cdot\frac{r_2+d-r_1-w}{2}\cdot\frac{r_1+w+d-r_2}{2}\cdot\frac{r_1+w+r_2-d}{2}}\\
\mu_2=\sqrt{\frac{r_1+r_2+w+d}{2}\cdot\frac{r_2+w+d-r_1}{2}\cdot\frac{r_1+d-r_2-w}{2}\cdot\frac{r_1+r_2+w-d}{2}}.
\end{cases}
\]
So
\[
\frac{\mu_1}{\mu_2}=\sqrt{\frac{(r_2+d-r_1-w)(r_1+w+d-r_2)}{(r_2+w+d-r_1)(r_1+d-r_2-w)}}.
\]
Observe that $g=r_1+d-r_2$ and so
\[
\frac{\mu_1}{\mu_2}=\sqrt{\frac{(2d-g-w)(g+w)}{(2d-g+w)(g-w)}}=\sqrt{\frac{(X+2w)Y}{X(Y+2w)}}=\sqrt{1+\frac{2w(Y-X)}{XY+2wX}},
\]
where $X=g-w$ and $Y=2d-g-w$.
Since in this case, $d\ge r_2-r_1+2w\ge 3w$, and so $g=r_1+d-r_2\ge 2w$.
This implies $X=g-w=\Theta(g)$.
On the other hand, $Y-X=2(d-g)=2(r_2-r_1)>0$ and also we have $3w\le d \le Y$.
Therefore,
\[
\frac{\mu_1}{\mu_2}=\sqrt{1+O(\frac{w}{g})}.
\]
Note that
\[
\begin{cases}
\mu_1=\frac{h_1d}{2}\\
\mu_2=\frac{h_2d}{2}
\end{cases}.
\]
Since $2w\le g$, we therefore have
\[
\frac{h_1}{h_2}=\frac{\mu_1}{\mu_2}=\sqrt{1+O(\frac{w}{g})}=1+O(\frac{w}{g})\implies{\frac{h_2}{h_1}=1-O(\frac{w}{g})}.
\]
By Heron's formula,
we obtain
\[
h_1=O\left(n\sqrt{\frac{g}{d}}\right).
\]
So
\begin{align}
\label{eq:dt}
|DT|=y_{BD}=h_1-h_2=h_1\left(1-\frac{h_2}{h_1}\right)=h_1\cdot O\left(\frac{w}{g}\right)=O\left(n\sqrt{\frac{w^2}{dg}}\right).
\end{align}
Therefore, by Equation~\ref{eq:bt} and Equation~\ref{eq:dt}, we have
\[
|BD|\le|BT|+|DT|=O\left(n\sqrt{\frac{w^2}{dg}}\right).
\qedhere
\]
\end{proof}

\section{Proof of Lemma~\ref{lem:num-int-ring}}
\label{sec:proof-num-int-ring}
\numintring*

\begin{proof}
We consider for the sake of contradiction that in our construction
for every set $\sS$ of $l=cw^2/\sqrt{T}$ annuli for any positive constant $c$,
we cannot find two annuli in $\sS$ such that their intersection area is $O(nw\sqrt{1/T})$.

First observe that by our construction, any two annuli from the same family have zero intersection.
Also the distance between the centers of any two annuli is less than the radius of any annulus.
Furthermore, the width of an annulus is always less than the distance between the centers of any two annuli.
So by Lemma~\ref{lem:general-ring-int},
the intersection area of any two annuli in our construction with distance $\Omega(wT)$ is $O(nw\sqrt{1/T})$.
So for our assumption to hold,
the distance between the centers of any two annuli in $\sS$ is $o(wT)$.

Let $P$ be a point in the intersection of $\sS$, then $P$ is contained in every annulus in $\sS$.
Now consider an arbitrary annulus $\rR_1 \in \sS$ centered at $O_1$
and another annulus $\rR_2 \in \sS$ centered at $O_2$ for some $O_2$ not in line $PO_1$.
Connect $PO_1$ and $PO_2$,
for $\rR_1, \rR_2$ to contain $P$,
we must have $|PO_1| = r_1+a$ and $|PO_2| = r_2+b$ for $0 \le a, b \le w$ 
as shown in Figure~\ref{fig:ap-small-d-int-region}.

We first consider the case when the distance between the centers of $\rR_1$ and $\rR_2$ is
no more than $r_2-r_1+2w$.
In this case, their intersection area
is upper bounded by $O(nw\sqrt{w/d})$ according to Lemma~\ref{lem:general-ring-int}.

\begin{figure}[h]
  \centering
  \includegraphics[scale=0.4]{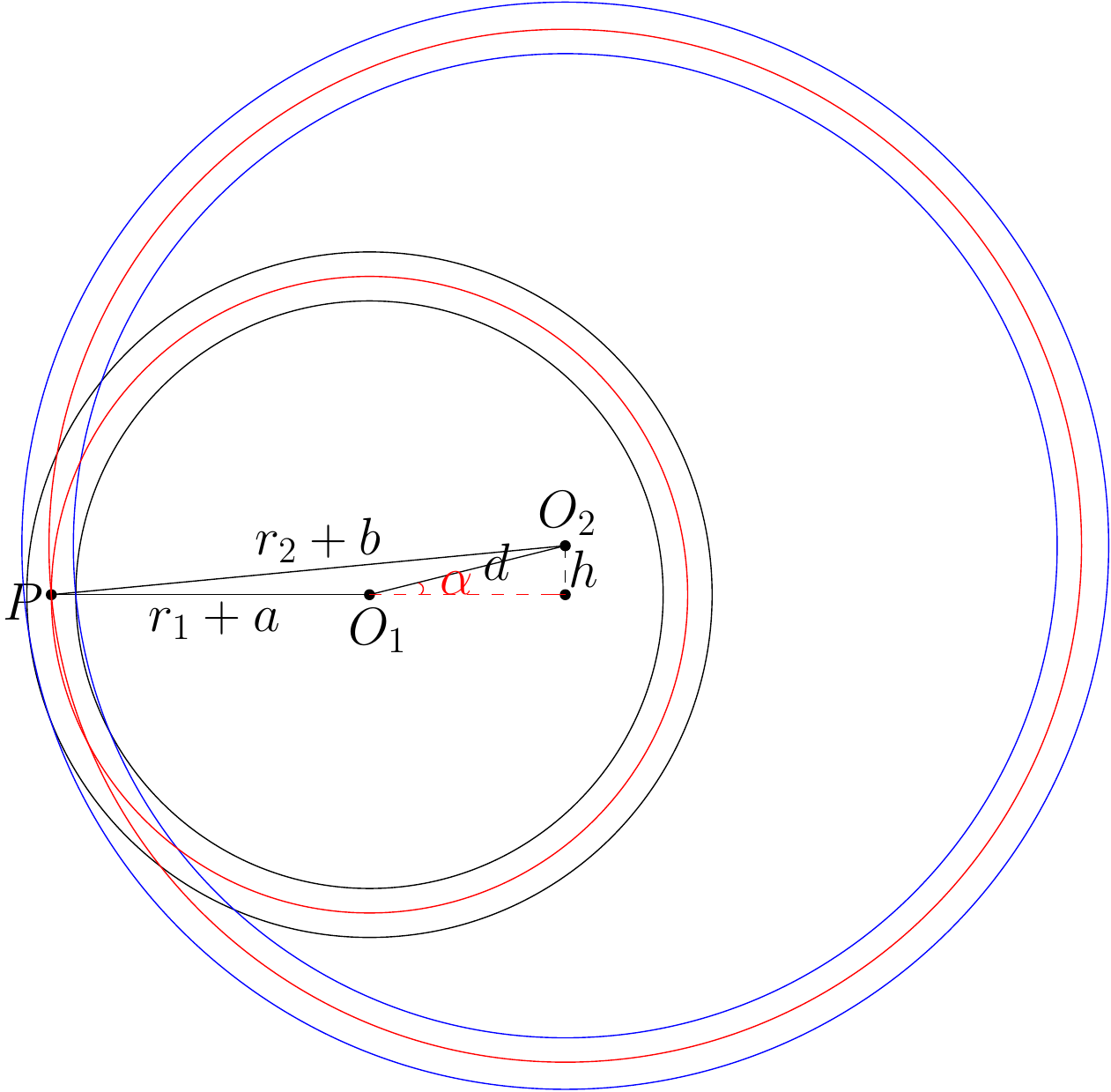}
  \caption{Intersection of Two Annuli When Distance Between Centers is Small}
  \label{fig:ap-small-d-int-region}
\end{figure}

By Heron's formula,
\begin{align*}
Area_{\triangle{PO_1O_2}}
&=\sqrt{\frac{r_1+a+r_2+b+d}{2}\cdot\frac{r_2-r_1+b-a+d}{2}\cdot\frac{r_1+a-r_2-b+d}{2}\cdot\frac{r_1+a+r_2+b-d}{2}}\\
&=O(r_1\sqrt{dw})\\
&=\Theta(r_1h),
\end{align*}
where the second equality follows from $a,b\le w\le d$ and $r_2-r_1-w\le d\le r_2-r_1+2w$,
and $h$ is the distance between $O_2$ and line $PO_1$.
This implies $h=O(\sqrt{dw})$.
Since $d=o(wT)$,
We obtain that 
\begin{align}
\label{eq:h-range-small-d}
h =o(w\sqrt{T}).
\end{align}

Now we consider the case when the distance between the centers of $\rR_1$ and $\rR_2$ is
more than $r_2-r_1+2w$. See Figure~\ref{fig:ap-bound-region} for an example.

\begin{figure}[h]
  \centering
  \includegraphics[scale=0.4]{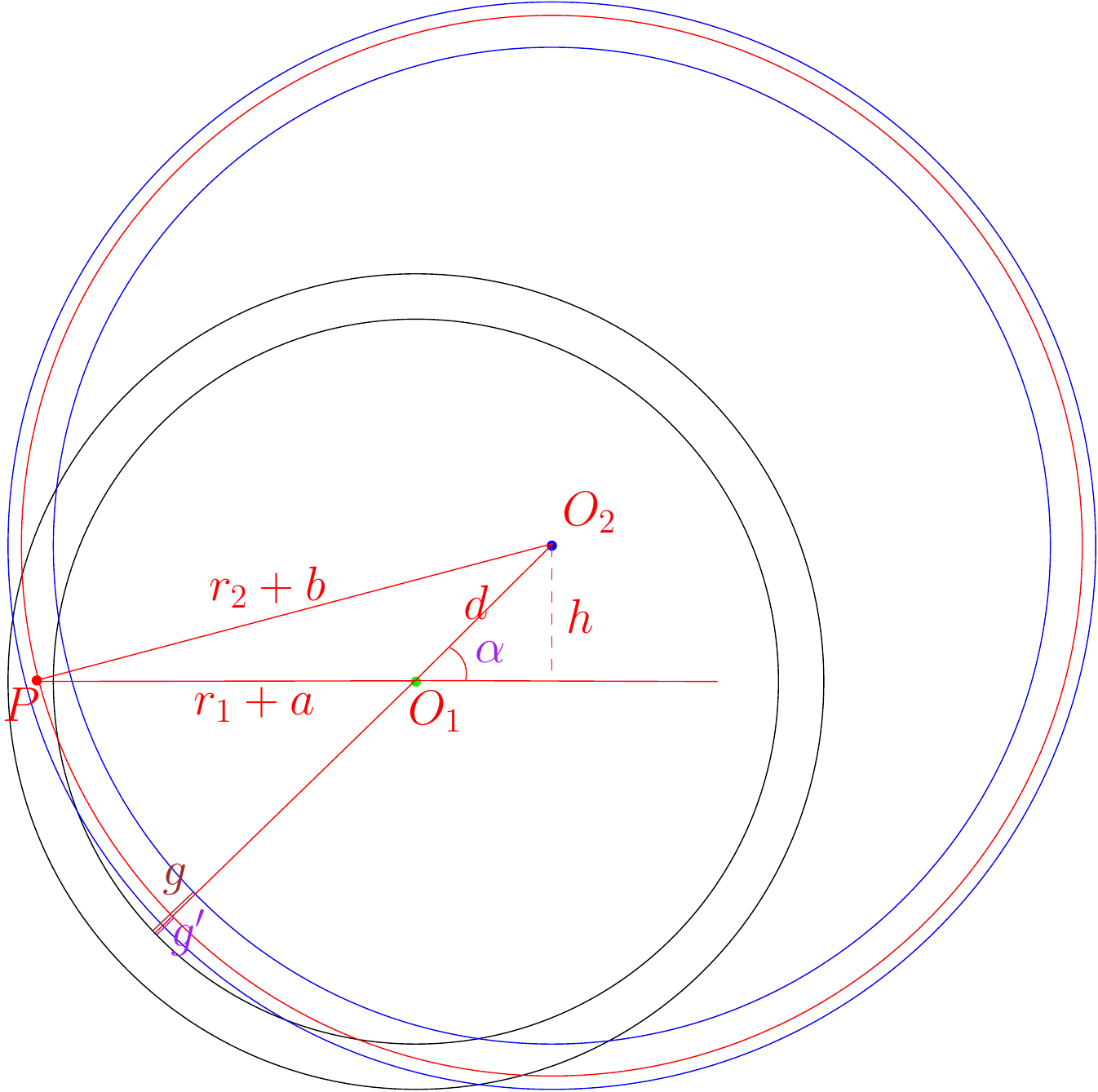}
  \caption{Intersection of Two Annuli When Distance Between Centers is Large}
  \label{fig:ap-bound-region}
\end{figure}

Let $g=r_1+d-r_2$ be the distance between the inner circle of $\rR_1$ and 
the inner circle of $\rR_2$.
Let $\alpha$ be the angle between $O_1O_2$ abd $PO_1$.
Note that $\alpha>0$ since $O_2$ is not in $PO_1$.
W.l.o.g., we assume $0<\alpha \le \pi/2$.
The situations for other values of $\alpha$ are symmetric.
Then
\begin{align}
\label{eq:alpha}
\cos(\pi-\alpha)=\frac{(r_1+a)^2+d^2-(r_2+b)^2}{2d(r_1+a)}.
\end{align}
Since $a, b<g=r_1+d-r_2<d<r_2$,
we obtain from~(\ref{eq:alpha}) that
\begin{align}
\label{eq:g-bound}
g=\Theta(d(1-\cos\alpha)). 
\end{align}
So according to Lemma~\ref{lem:general-ring-int} the intersection area of $\rR_1, \rR_2$ is upper bounded by
\[
A=O\left(wn\sqrt{\frac{w^2}{dg}}\right).
\]
Let $A =\omega\left(nw\sqrt{\frac{1}{T}}\right)$, by equation~(\ref{eq:g-bound}), we get
\begin{align}
\label{eq:d-range}
d^2(1-\cos\alpha) = o(w^2T),
\end{align}
Since $d=\frac{h}{\sin\alpha}=\frac{h}{\sqrt{1-\cos^2\alpha}}$, 
and $0<\alpha \le \pi/2$, we plug $d$ in inequality~(\ref{eq:d-range}) and obtain
\[
\frac{h^2}{1-\cos^2\alpha}(1-\cos\alpha)=\frac{h^2}{1+\cos\alpha} = o(w^2T).
\]
This implies
\begin{align}
\label{eq:h-range-big-d}
h=o\left(\sqrt{w^2T(1+\cos\alpha)}\right)=o\left(w\sqrt{T}\right).
\end{align}

So in order to have no two annuli having intersection area $O(nw\sqrt{1/T})$,
the distance between the centers of annuli in $\sS$ and $PO_1$ must be $o(w\sqrt{T})$.
We have already shown that the distance between any two centers is $o(wT)$.
Together they imply that the centers of $\sS$ fit in a rectangle of shape $o(wT)\times o(w\sqrt{T})$.
However, in our construction, each grid cell is of size $T\times T$,
this implies we only have $o(w^2/\sqrt{T})$ centers in the rectangle.
But we should have $cw^2/\sqrt{T}$ centers, a contradtion.
\end{proof}

\end{document}